\documentclass[a4paper]{article}

\usepackage{article_preamble}
\usepackage{kbordermatrix}
\usepackage[ruled]{algorithm2e}

\newcommand{\Din}{\mathcal{D}_{\text{in}}}
\newcommand{\Dout}{\mathcal{D}_{\text{out}}}
\newcommand{\Dother}{\mathcal{R}}
\newcommand{\Cgr}{\mathcal{C}}

\title{Fast tree inference with weighted fusion penalties}

\author{
  Julien Chiquet, Pierre Gutierrez and Guillem Rigaill
}

\begin{document}

\maketitle

\begin{abstract}
  Given a  data set with many  features observed in a  large number of
  conditions, it is  desirable to fuse and  aggregate conditions which
  are  similar  to  ease  the  interpretation  and  extract  the  main
  characteristics of the data.  This paper presents a multidimensional
  fusion penalty framework to address this question when the number of
  conditions  is  large.  If  the  fusion  penalty  is encoded  by  an
  $\ell_q$-norm,  we  prove  for  uniform weights  that  the  path  of
  solutions is a tree which is suitable for interpretability.  For the
  $\ell_1$ and  $\ell_\infty$-norms, the path is  piecewise linear and
  we derive  a homotopy  algorithm to recover  exactly the  whole tree
  structure.  For  weighted $\ell_1$-fusion penalties,  we demonstrate
  that distance-decreasing  weights lead to balanced  tree structures.
  For  a  subclass  of  these weights  that  we  call  ``exponentially
  adaptive'', we derive  an $\mathcal{O}(n\log(n))$ homotopy algorithm
  and we prove an asymptotic oracle property.  This guarantees that we
  recover the underlying structure of the data efficiently both from a
  statistical and  a computational point  of view.  We provide  a fast
  implementation  of the  homotopy  algorithm for  the single  feature
  case, as  well as  an efficient embedded  cross-validation procedure
  that takes advantage of the tree structure of the path of solutions.
  Our  proposal outperforms  its competing  procedures on  simulations
  both in terms of timings and  prediction accuracy.  As an example we
  consider  phenotypic   data:  given   one  or  several   traits,  we
  reconstruct a balanced tree structure  and assess its agreement with
  the known taxonomy.
\end{abstract}

\section{Introduction}
\label{sec:intro}

As data floods in, it is  now possible to compare many features across
a very  large number of conditions  in various fields of  science.  To
cite  but  a  few,  we   encounter  this  setting  in  genomics  where
high-throughput technologies allow us  to monitor the expression level
of many genes  (the features) at various stages of  a given biological
process  (the conditions);  this  also occurs  in phylogenetics  where
several  quantitative traits  (the  features) are  available for  many
species (the  conditions).  Beyond  biological sciences, sets  of data
gathered  in  astronomy are  now  routinely  composed of  millions  of
conditions for  hundreds of features.   An interesting question  is to
group  together --  or fuse  --  these conditions  across the  feature
space, arguing that  these conditions should not  really be considered
as different.  In  other words, we aim at  recovering an interpretable
clustering of those conditions.

There are basically two cases:  either a prior group structure between
the conditions is known, or it is not. %
In the  first case, one typically  applies one-way ANOVA --  or MANOVA
for  multiple  features --  to  test  for any  significant  difference
between  all  pairs  of  groups.   The  final  structure  between  the
conditions then depends on the level  of significance used to test for
differences.  However, when  the number $K$ of groups  is large, which
typically occurs for a large number $n$ of conditions, this leads to a
multiple-testing issue  and algorithmic  problems since the  number of
pairwise tests  is in  $\mathcal{O}(K^2)$.  Furthermore, each  test is
performed independently and the resulting structure
is   not  necessarily  simple   and  easily
interpretable.

In the  second case,  when no prior  group structure is  available, we
basically face a clustering problem over the multidimensional space of
the  features.    A  popular  heuristic  to  solve   this  problem  is
agglomerative  clustering,  which  defines  a  hierarchical  structure
between   the  conditions.    Hierarchies  are   very   appealing  for
interpretability.  A serious  bottleneck  of agglomerative  clustering
when   analyzing    large   data    sets   is   its    complexity   in
$\mathcal{O}(n^3)$, which  can be reduced  to $\mathcal{O}(n^2)$ using
single-linkage clustering.

There are two major issues for large  values of $n$: $i)$ the need for
an interpretable structure  between the conditions and  $ii)$ the need
for   a  computationally   and   statistically  efficient   estimation
procedure.  These two goals  cannot be reached simultaneously, neither
by  MANOVA   nor  by  agglomerative  clustering   algorithms,  due  to
restrictions either on the  interpretability of the inferred structure
or the computational  burden of the procedure.  This  paper presents a
unifying approach to tackle these two problems simultaneously by means
of a weighted fusion penalty  that constructs a hierarchical structure
on the  conditions at a  low computational  cost and reaching  the two
aforementioned goals.   Section \ref{sec:model} presents  our proposal
in detail and  puts it in perspective with existing  methods.  Then we
use the optimality conditions detailed in Section \ref{sec:optimality}
to     characterize      the     regularization      path     (Section
\ref{sec:pathANDtree}).   In  Section  \ref{sec:weights},  we  propose
weights  for which  the  path  is provably  without  splits.  For  the
$\ell_1$-norm some of those weights  lead to a desirable balanced tree
structure. In Section \ref{sec:optim}  we present a homotopy algorithm
which is in $\mathcal{O}(K \log K)$  for well chosen weights.  We also
provide an  efficient embedded  cross-validation procedure to  tune up
the level of  aggregation -- or fusion -- between  groups in the ANOVA
settings.  Numerical experiments  illustrate the extremely competitive
performance   of  our   algorithm  in   terms  of   timings.   Section
\ref{sec:consistency}   presents   consistency  results   that   bring
statistical guarantees for our approach.  We illustrate our theorem on
a simulation study that shows that our weights are more efficient than
those    of    its    competing    procedures.     Finally,    Section
\ref{sec:application}   is  dedicated   to  a   complete  example   in
phylogenetics where our  method is applied to the  reconstruction of a
balanced tree  structure across several phylogenetic  features between
many species.   We assess its  relevance by comparison with  the known
phylogeny.


\section{A penalized framework for tree inference}
\label{sec:model}

To  bring MANOVA  and  hierarchical clustering  together  in the  same
unifying penalized framework,  note that the latter  can be considered
as a  particular case of the  former when there is  only one condition
per group,  \textit{i.e}, when  $K=n$.  This  can be  thought of  as a
non-informative prior on the clustering between the conditions.

To be more  specific, we set $y_{ij}$ the observation  of a continuous
random variable that  describes the intensity of the  $j$th feature in
condition $i$,  with $i\in\set{1,\dots,n}$  and $j\in\set{1,\dots,p}$.
The $p$-dimensional vector $\by_i=(y_{i1}, \dots, y_{ip})$ encompasses
the data  related to  condition $i$  across the  $p$ features.  We are
given a partition with $K$ groups  as prior knowledge that is depicted
by   the    indexing   function   $\kappa   :    \set{1,\dots,n}   \to
\set{1,\dots,K}$.   In words,  $\kappa$ indicates  the group  to which
condition $i$ is allocated \emph{a priori}.  The number of elements
in  group $k$  is denoted  by $n_k  = \text{card}\set{i:\kappa(i)=k}$,
such that $\sum_k n_k = n$.

One-way  MANOVA  is a  multivariate  linear  regression problem  whose
parameters  are fitted  by  minimizing the  residual  sum of  squares,
\textit{i.e.},
\begin{equation*}
  \minimize_{\bbeta    \in    \Rset^{Kp}}   \sum_{i=1}^n    \sum_{j=1}^p
  \left(      y_{ij}     -     \beta_{\kappa(i)j}      \right)^2     =
  \argmin_{\bbeta\in\Rset^{Kp}}   \sum_{i=1}^n    \left\|   \by_i   -
    \bbeta_{\kappa(i)} \right\|_2^2 ,
\end{equation*}
where $\beta_{kj}$  is the  coefficient for the  $j$th feature  in the
$k$th        group,        such         that        $\bbeta_k        =
(\beta_{k1},\dots,\beta_{kp})\in\Rset^p$.  The final structure between
the  conditions is  obtained  by testing  for significant  differences
between        all        pairs       of        estimated        means
$(\hat{\beta}_{kj},\hat{\beta}_{\ell j})$ using Fisher statistics.

Compared to MANOVA, hierarchical clustering assumes one individual per
group,  that  is  $K=n$  or  equivalently  $\kappa(i)  =  i$  for  all
$i=1,\dots,n$.  It  performs agglomeration by recursively  joining the
closest   points.    As   suggested   by   \cite{HOCKING-clusterpath},
hierarchical  clustering aims  at solving  the following  optimization
problem:
\begin{equation}
  \label{eq:hierarchical_clustering}
  \minimize_{\bbeta\in\Rset^{np}} \sum_{i=1}^n  \left\| \by_i - \bbeta_i
  \right\|_2^2, \qquad \text{ s.t. } \sum_{i>i'}
  \I_{\bbeta_i \neq \bbeta_i'} \leq t.
\end{equation}
The complete hierarchy between the conditions is recovered by starting
from $t=n(n-1)/2$, where no constraint applies, then by decreasing $t$
until   all  points   agglomerate.    This   immediately  suggests   a
corresponding scheme for agglomerating  groups of conditions in MANOVA
just by using the prior grouping  knowledge encoded by $\kappa$ in the
square loss.  However,  Problem \eqref{eq:hierarchical_clustering} and
its  MANOVA  counterpart  are   difficult  combinatorial  problems  in
general.   To overcome  this  restriction, we  consider the  following
convexified Lagrangian formulation which  includes the whole family of
optimization problems discussed throughout this paper:
\begin{equation}
  \label{eq:criterion_general}
  \minimize_{\bbeta \in \Rset^{Kp}} \frac{1}{2} \sum_{i=1}^n \left\| \by_i -
    \bbeta_{\kappa(i)} \right\|_2^2
  + \lambda \ \sum_{k,\ell:k\neq\ell} w_{k\ell} \ \Omega( \bbeta_k - \bbeta_\ell).
\end{equation}
In general, $\Omega$ is a norm and $w_{k\ell}$ are positive, symmetric
weights  over  all pairs  of  groups  in $\set{1,\dots,K}$  such  that
$w_{k\ell}>0$  and $w_{k\ell}=w_{\ell  k}$. The  penalty term  and the
choice of  $\Omega$ is designed  to encourage elements of  $\bbeta$ to
``fuse''   by   enforcing   similarity  between   pairs   of   vectors
$(\bbeta_k,\bbeta_\ell)$  as in  the  fused-Lasso signal  approximator
\citep{friedman2007pathwise},  which   is  an   $\ell_1$-based  method
designed to  aggregate pairs of  elements.  As  such, we refer  to the
penalty term in~\eqref{eq:criterion_general}  as a ``fusion'' penalty.
In the  multidimensional case though,  other choices are  possible for
$\Omega$ that induce a fusion effect.  The level of fusion is tuned by
two parameters:  the global level  of penalty $\lambda$ and  the group
specific weights  $w_{k\ell}$, the choice  of which is of  the highest
importance. It conditions both $i)$ the ability of the method to infer
an interpretable structure between the conditions, $ii)$ the existence
of fast algorithms  to fit the parameters $\bbeta$  for various values
of $\lambda$  and $iii)$ the  existence of statistical  guarantees for
the estimator.  The  main objective of this paper is  to study classes
of weights that reach these three goals simultaneously.

\paragraph*{Links            to            existing            works.}
Problem~\eqref{eq:criterion_general}  is   a  generalization   of  two
interesting existing procedures related to  ours.  The first one arose
in the  clustering framework  and is  known as  the \emph{Clusterpath}
\citep{HOCKING-clusterpath}.  The  \emph{Clusterpath} covers  cases in
\eqref{eq:criterion_general} where $K=n$ and $\Omega(\cdot) = \| \cdot
\|_q$ with $q\in\set{1,2,\infty}$.  Still, for general weights, the complexity
of the associated  algorithms does not improve  over the agglomerative
clustering, and the  inferred structure is not a  tree.  However, when
$q=1$, the path of solutions is linear with respect to $\lambda$ and a
homotopy  algorithm  is  used by  \citeauthor{HOCKING-clusterpath}  to
recover the  solutions over  all the values  of $\lambda$  that either
correspond   to  events   of  fusion   or  split   between  a   couple
$(\bbeta_k,\bbeta_\ell)$.  Moreover, if  $w_{k\ell}=1$ and $q=1$, they
showed that no split event can occur and that a homotopy algorithm can
be  implemented  in $\mathcal{O}(n  \log(n))$.   In  other words,  the
reconstructed structure is  a tree in this case.   However the unitary
weights typically lead  to unbalanced hierarchies which  are not fully
satisfactory.

A  second close  cousin to  our  approach is  the \emph{Cas-ANOVA}  of
\cite{bondell2008simultaneous}.       \emph{Cas-ANOVA}       is      a
$\ell_1$-penalized  version   of  the   ANOVA  which   corresponds  to
\eqref{eq:criterion_general} in the univariate setting where $p=1$ and
$\Omega(\cdot)  =  \| \cdot  \|_1$.   The  main contribution  of  this
proposal    is    statistical:    \citeauthor{bondell2008simultaneous}
introduce   adaptive   weights    $w_{kl}   \varpropto   \sqrt{n_k   +
  n_\ell}/(\bar{y}_{k}-\bar{y}_{\ell})$, where $n_k$  is the number of
conditions in group $k$ and $\bar{y}_k = \sum_{i:\kappa(i)=k} y_i/n_k$
is  the  corresponding  empirical  mean.  Similar  weights  have  been
proposed   in   \cite{gertheiss2010sparse}   to  cope   with   ordered
categorical variables.   These weights have an  adaptive property such
that  the  corresponding  estimator   of  $\bbeta$  enjoys  asymptotic
consistency,    in    the    manner     of    the    adaptive    Lasso
\citep{zou2006adaptive}. Still,  \emph{Cas-ANOVA} weights do  not lead
to a tree when the number  of individuals per condition is unbalanced,
\textit{i.e.}, $n_k \neq n_\ell$ for any couple $(k,\ell)$.  Moreover,
the optimization procedure is  in $\mathcal{O}(K^2)$ and only provides
the solution  for a  given $\lambda$.   We also  experienced numerical
instability using \emph{Cas-ANOVA} weights.

\paragraph*{Contributions.}   Compared  to   these   two  works,   our
contributions are the following:
\begin{itemize}
\item We prove that no split can  occur along the path of solutions in
  \eqref{eq:criterion_general} when  $w_{k\ell} = n_k  \cdot n_{\ell}$
  and  $\Omega(\cdot)$ is  an $\ell_q$-norm.   As a  consequence, this
  proves  that  the  \emph{Clusterpath}  does not  split  for  unitary
  weights,  whatever  the  choice  of  the  norm  (as  conjectured  by
  \citeauthor{HOCKING-clusterpath} for the $\ell_2$-norm).
\item  When Problem  \eqref{eq:criterion_general} is  separable across
  the features (\textit{e.g.}, when $\Omega$ is the $\ell_1$-norm), we
  introduce distance-decreasing  weights for  which we prove  that the
  path is a  tree.  From an interpretation point of  view, this family
  of   weights   is   particularly   interesting  as   it   leads   to
  \emph{balanced} tree structures.
\item  For  the  $\ell_1$-norm, we  introduce  exponentially  adaptive
  weights that  enter the family of  distance-decreasing weights. They
  enjoy asymptotic  oracle properties that guarantee  selection of the
  true underlying structure  for a large scale  of possible $\lambda$.
  This shows that our estimator  shares the same asymptotic properties
  as \emph{Cas-ANOVA},  but for a larger  range of $\lambda$ and  at a
  much lower computational cost.
\item    We    provide    a    general    homotopy    algorithm    for
  \eqref{eq:criterion_general}  when  $\Omega(\cdot)  =  \|\cdot\|_1$,
  whatever  the  choice of  $w_{k\ell}$.   On  a single  feature,  the
  initialization for unspecified weights  is in $\mathcal{O}(K^2)$ and
  the  homotopy itself  is in  $\mathcal{O}(K \log(K))$.   However, we
  propose a faster initialization procedure for exponentially adaptive
  weights  such that  the  whole  complexity for  $p$  features is  in
  $\mathcal{O}(p K \log(K))$ -- or  $\mathcal{O}(p n \log (n))$ in the
  clustering framework.
\item  When  the number  $K$  of  prior  groups  is smaller  than  $n$
  (\textit{e.g.},  in   the  ANOVA  settings,  when   there  are  some
  replicates  per condition/group),  a  natural cross-validation  (CV)
  error can  be defined.  In  this case,  we develop a  fast procedure
  that takes advantage  of the DAG (directed  acyclic graph) structure
  of the path of solutions along $\lambda$.  This approach has a lower
  complexity than a standard CV procedure.
\end{itemize}

In     short,     we     propose    choices     for     weights     in
\eqref{eq:criterion_general}  that induce  a  balanced tree  structure
between the  conditions such that the  associated estimation procedure
enjoys     the     good     computational    properties     of     the
$\ell_1$-\emph{Clusterpath}  with   unitary  weights,   with  stronger
statistical guarantees than \emph{Cas-ANOVA}.

\paragraph*{Motivating example in phylogeny.}   As a simple motivating
example, we  consider a  univariate problem in  phylogeny. We  want to
reconstruct a tree between many  species based on some simple features
(like the height,  or the weight of individuals).   Ideally, this tree
should resemble the  known phylogeny.  We illustrate this  task on the
``Animal  Ageing Longevity  Database''\footnote{publicly available  at
  \url{http://genomics.senescence.info/species/}},    which   provides
various  features   for  many  animal  species.    Here,  we  consider
classifying  bird species  based  on their  birth  weight.  The  known
phylogeny groups  these $n=184$  individuals into $40$  bird families,
themselves grouped into $15$ orders.  We reconstruct the tree based on
the weights  and check whether  it matches  the orders and  the family
classification.         Recovered       solution        paths       of
\eqref{eq:criterion_general}      are      plotted      in      Figure
\ref{fig:phylo_trees}     for    $a)$     the    Cas-ANOVA     weights
\citep{bondell2008simultaneous}  ;  $b)$ the  ``default''  Clusterpath
weights \citep{HOCKING-clusterpath}; and $c)$  our own weights that we
call ``fused-ANOVA'' weights.   On the left panel,  the Cas-ANOVA path
includes many  splits which make interpretation  rather difficult.  On
the middle panel, default Clusterpath  weights, as expected, provide a
tree structure.  Still,  the structure of this tree  is unbalanced and
thus not fully satisfactory in the  sense that small groups often fuse
with  very large  ones. Specifically,  the Clusterpath  tree does  not
capture  the  simple   fact  that  there  are   visibly  three  groups
corresponding  to  light,  medium  or heavy  birds.   Conversely,  the
fused-ANOVA  tree in  the right  panel  is more  balanced and  clearly
exhibits these three  groups.  Furthermore, it is  in better agreement
with the  known phylogenetic classification, improving  the rand index
by $5\%$ compared to ClusterPath.
\begin{figure}[htbp!]
  \centering
  \noindent\begin{tabular}{@{}c@{}c@{}c@{}cc@{}}
    \rotatebox{90}{\hspace{.25em}estimated coefficients $\hatbbeta_{\lambda}$} 
    & \includegraphics[width=.275\textwidth]{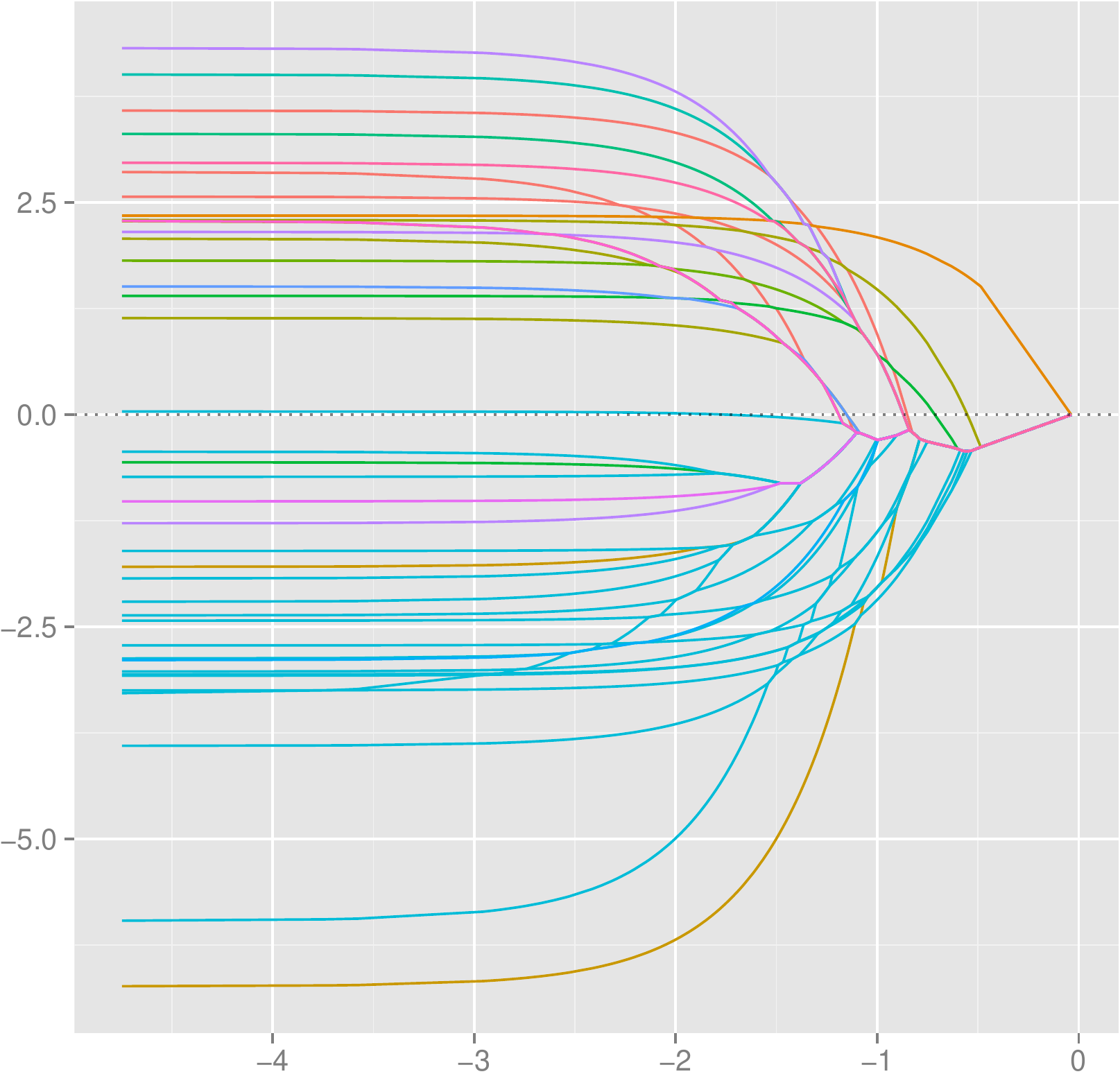}
    & \includegraphics[width=.275\textwidth]{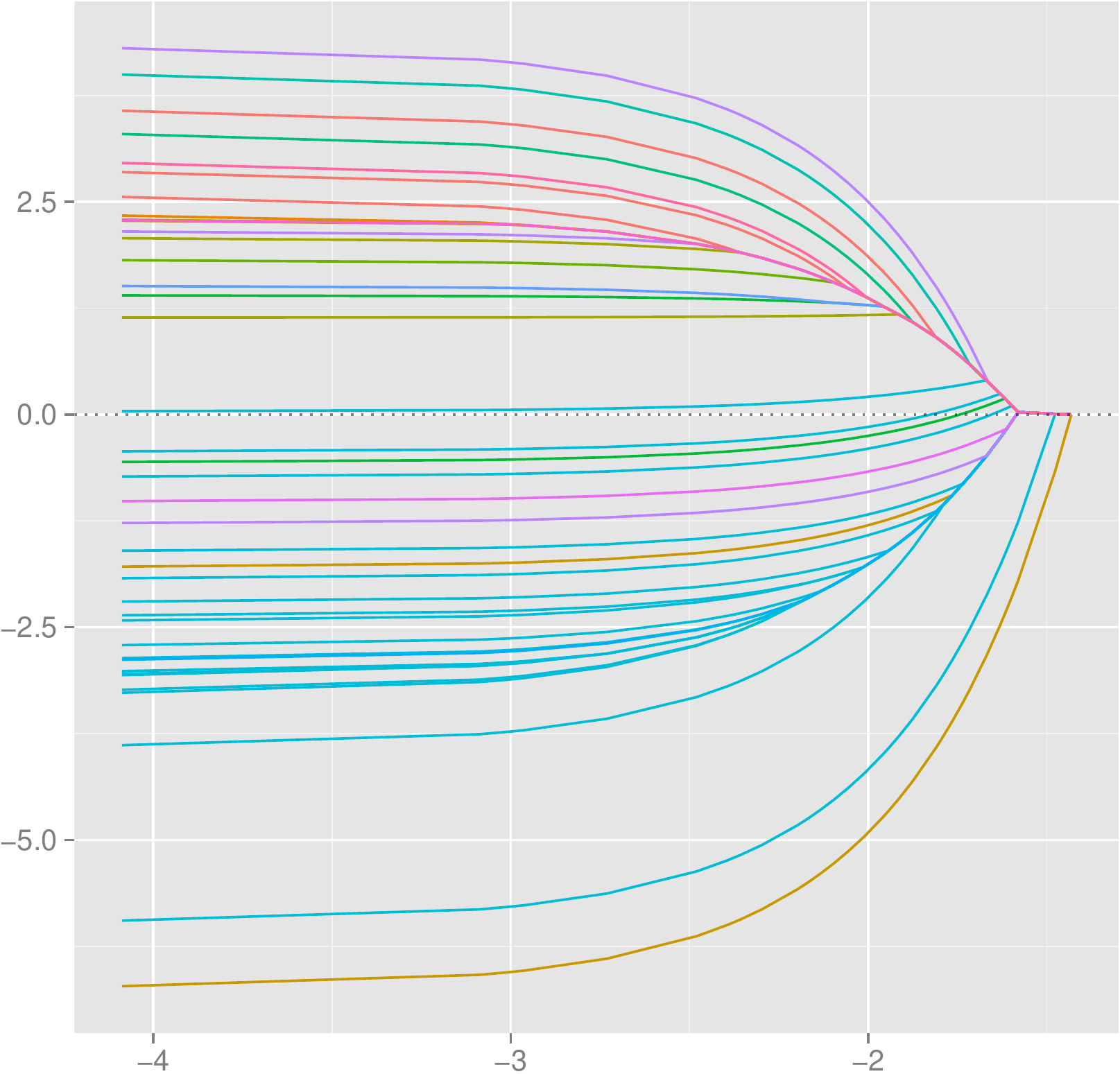}
    & \includegraphics[width=.275\textwidth]{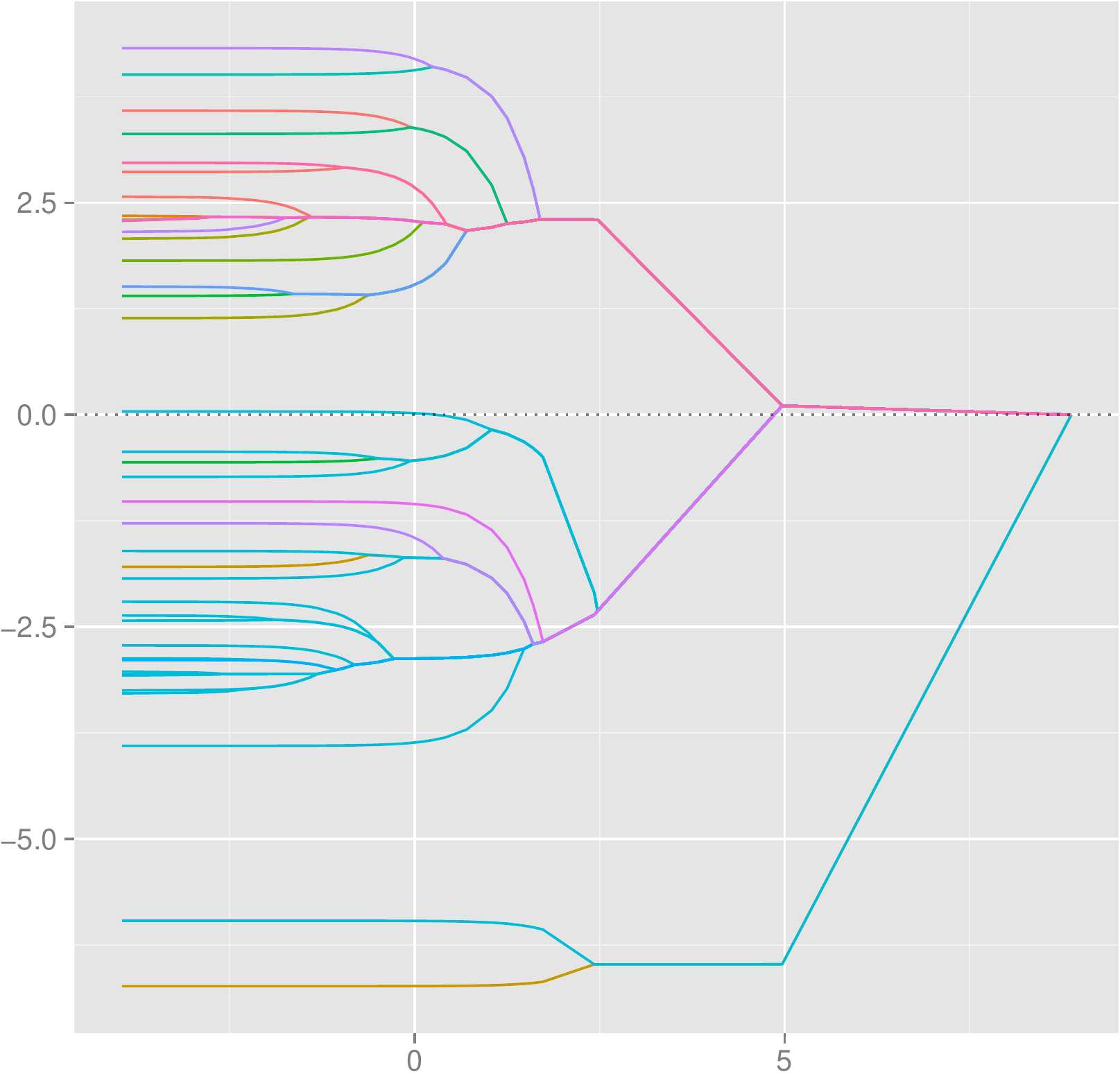} 
    & \includegraphics[clip=true, trim= 0 -13pt 0 0,width=.072\textwidth]{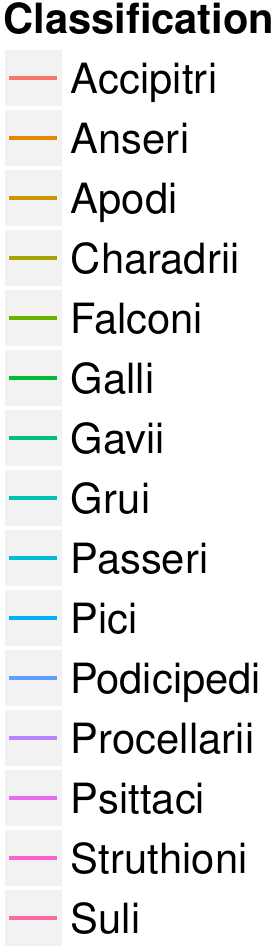} \\
    & \multicolumn{4}{c}{tuning parameter $\lambda$ (log scale)} \\
    & $a)$ Cas-ANOVA & $b)$ Clusterpath & $c)$ Fused-ANOVA & \\
  \end{tabular}
  \caption{Reconstructed  phylogenetic  trees  for  various  weighting
    schemes.   Families classified in  the same  order share  the same
    color.}
  \label{fig:phylo_trees}
\end{figure}

\paragraph*{Multidimensional $\ell_1$ Clusterpath and fused-ANOVA.}

In the previous  example, we consider only one  feature.  In practice,
one often has to consider multiple  features at the same time. This is
possible with our proposed weighted $\ell_1$-penalty. Indeed, as noted
by \cite{HOCKING-clusterpath}, Problem~\eqref{eq:criterion_general} is
separable on  dimensions when considering the  $\ell_1$-penalty, which
is  also  the  case  for   our  weighted  fused-ANOVA  scheme.   Thus,
Clusterpath  and  fused-ANOVA  algorithms solve  the  multidimensional
problem in two steps:
\begin{enumerate}
\item   First,   they  recover   $p$   independent   trees  (one   per
  dimension). This task can be easily executed in parallel.
\item Second they aggregate those $p$  trees in a consensus tree. This
  is  done   by  considering   the  same  penalty   value  ($\lambda$)
  corresponding to a given height  in those trees. Two individuals $k$
  and $\ell$  are in  the same multidimensional  cluster if  they have
  been fused on every dimension.
\end{enumerate}

This  multidimensional  classification  is  recovered  on  a  grid  of
$\lambda$   in   the   \textbf{clusterpath}   package   and   in   the
\textbf{fusedanova} package. 

Note however that the classification recovered over all the dimensions
is not necessarily better than  those recovered on single, well-chosen
features.     We   illustrate    this    point   at    the   end    in
Section~\ref{sec:application}  on phylogenetic  data: in  a number  of
cases, the  best agreement with the  known phylogeny is obtained  by a
single-feature-based tree.


\section{Optimality conditions and consequences}
\label{sec:optimality}

We  start  by   characterizing  Problem  \eqref{eq:criterion_general},
giving elementary facts which are at the basis of most of our results.
Note that the objective  function in \eqref{eq:criterion_general} is a
nonsmooth  function which  is  strictly convex  in  $\bbeta$ and  thus
admits a unique  solution when $\lambda\geq 0$.  This  solution can be
characterized by  the KKT (Karush-Kuhn-Tucker) conditions  that may be
derived     thanks     to     subgradient     calculus     \citep[see,
\textit{e.g.},][]{cvx_optim}.   In  the  case  at  hand,  $\bbeta$  is
optimal  if, for  all $k\in\set{1,\dots,K}$,  $\bbeta_k$ verifies  the
following subgradient equations:
\begin{equation}
  \label{eq:KKT_general}
  \bzr_p  = -  n_k (\bar{\by}_k  - \bbeta_k)  +  \lambda \sum_{\substack{\ell:
      \ell\neq k \\ \bbeta_k = \bbeta_\ell}}
  w_{k\ell} \boldsymbol\tau_{k\ell} + \lambda
  \sum_{\substack{\ell:
      \ell\neq k \\ \bbeta_k \neq \bbeta_\ell}}
  w_{k\ell} \frac{\partial \Omega (\bbeta_k - \bbeta_\ell)}{\partial \bbeta_k},
\end{equation}
where $\bar{\by}_k = \sum_{i:\kappa(i)=k}  \by_i/n_k$ is the vector of
empirical  means  for  the  $k$th group  across  every  feature.   The
$p$-dimensional vectors  $\boldsymbol\tau_{k\ell}$ are such  that, for
any $k$, there exists $\ell\neq  k$ with $\bbeta_k = \bbeta_\ell$ such
that   $\boldsymbol\tau_{k\ell}  =   -\boldsymbol\tau_{\ell  k}$   and
$\Omega(\boldsymbol\tau_{k\ell}) \leq 1$. We omit the proof as it is a
straightforward  adaptation of  the fused-Lasso  subgradient equations
\citep{hoefling2010path} to the multidimensional  case, with a general
norm $\Omega$.
\\

Interesting consequences arise  when summing the subgradient equations
\eqref{eq:KKT_general} for  all $\bbeta_k$ which are  ``fused'' in the
same cluster, as stated in the following Lemma.
\begin{lemma}
  \label{lem:KKT_sum_beta}
  Consider  a  cluster  $C=\set{k:\bbeta_k=\bbeta_C}$ formed  by  some
  $\bbeta_k$,     where     $\bbeta$     is    the     solution     to
  \eqref{eq:criterion_general}. Then we have
  \begin{equation}
    \label{eq:KKT_sum_beta}
    \bbeta_C  = \bar{\by}_C  -  \frac{\lambda}{n_C} \sum_{\ell\notin  C}
    w_{C\ell} \frac{\partial \Omega(\bbeta_C-\bbeta_\ell)}{\partial\bbeta_C},
  \end{equation}
  where  $n_C  =  \sum_{k\in  C}  n_k,  \bar{\by}_C  =  \sum_{k\in  C}
  \bar{\by}_k / n_C$ and $w_{C\ell} = \sum_{k\in C} w_{k\ell}$.
\end{lemma}

\begin{proof}
  By summing \eqref{eq:KKT_general} for all $k\in C$, we have
  \begin{equation*}
    \bzr_p   =  -   n_C   \bar{\by}_C  +   n_C   \bbeta_C  +   \lambda
    \sum_{k,\ell\in C: k\neq \ell} w_{k\ell} \boldsymbol\tau_{k\ell} + \lambda
    \sum_{k\in C, \ell\notin C} w_{k\ell} \frac{\partial \Omega (\bbeta_k - \bbeta_\ell)}{\partial \bbeta_k}.
  \end{equation*}
  Then, by the KKT conditions, we must have $\boldsymbol\tau_{k\ell} =
  -\boldsymbol\tau_{\ell k}$  for some  $k,\ell\in C$. Thus  the third
  term  on the  left-hand side  of  the above  expression vanishes  by
  symmetry  of the  weights $w_{k\ell}$.   Also notice  that $\partial
  \Omega(\bbeta_k   -   \bbeta_\ell)/\partial  \bbeta_k   =   \partial
  \Omega(\bbeta_{k'} - \bbeta_{\ell}) /  \partial \bbeta_{k'}$ for any
  $k,k'\in C, \ell\notin C$, and we easily get the desired result.
\end{proof}


\section{Regularization path and tree structure}
\label{sec:pathANDtree}

Characterization        of       the        minimization       Problem
\eqref{eq:criterion_general} in terms of  its optimality conditions is
essential in  many ways.  In particular,  Lemma \ref{lem:KKT_sum_beta}
allows  us  to  characterize  the  regularization  path  of  solutions
$\set{\bbeta(\lambda),  \lambda  > 0}$  depending  on  the choices  of
$\Omega$ and $w_{k\ell}$.  This is important for our problem since the
shape  of the  path is  actually the  structure recovered  between the
conditions.   This  is  also  important   since  it  may  induce  some
computational  properties  that  guarantee  a low  complexity  of  the
associated  fitting   procedure.   This  section   investigates  which
conditions  must be  imposed on  the regularization  path to  ensure a
structure  that is  fully satisfactory  both in  terms of  algorithmic
complexity and interpretability, namely, a balanced tree structure.

The  mildest  condition  which   is  required  is  continuity  of  the
regularization   path,    that   is   to   say,    of   the   function
$\set{\bbeta(\lambda),    \lambda   >    0}$:    without   continuity,
interpretability of the recovered structure is obviously out of reach.
This property is straightforward for solutions of problems of form
\eqref{eq:criterion_general} which is strictly convex.
However,  continuity  of  the  path   is  not  enough  to  provide  an
interpretable structure, and we  shall investigate conditions ensuring
that the  inferred structure  is a tree.   In terms  of regularization
path, it requires that any couple  of parameters which have fused at a
certain   time   $\lambda_0$    such   that   $\bbeta_k(\lambda_0)   =
\bbeta_\ell(\lambda_0)=\bbeta_C$  cannot  ``split''   anymore  in  the
future,  that is,  for  any  value $\lambda  >  \lambda_0$ that  would
correspond to a  higher level in the hierarchy of  the tree.  Insights
on this remark  can be found in Figure  \ref{fig:paths}, where various
regularization paths are plotted in  the univariate case. Paths on the
top and bottom left panels contain splits, while the remainders do not
and lead  to trees with different  shapes the properties of  which are
discussed later in this section.
\begin{figure}[htbp!]
  \centering
  \begin{tabular}{@{}l@{\hspace{.75em}}c@{}c@{}c@{}}
    \rotatebox{90}{\hspace{.6em}\small Piecewise quadratic path} &
    \includegraphics[width=.3\textwidth]{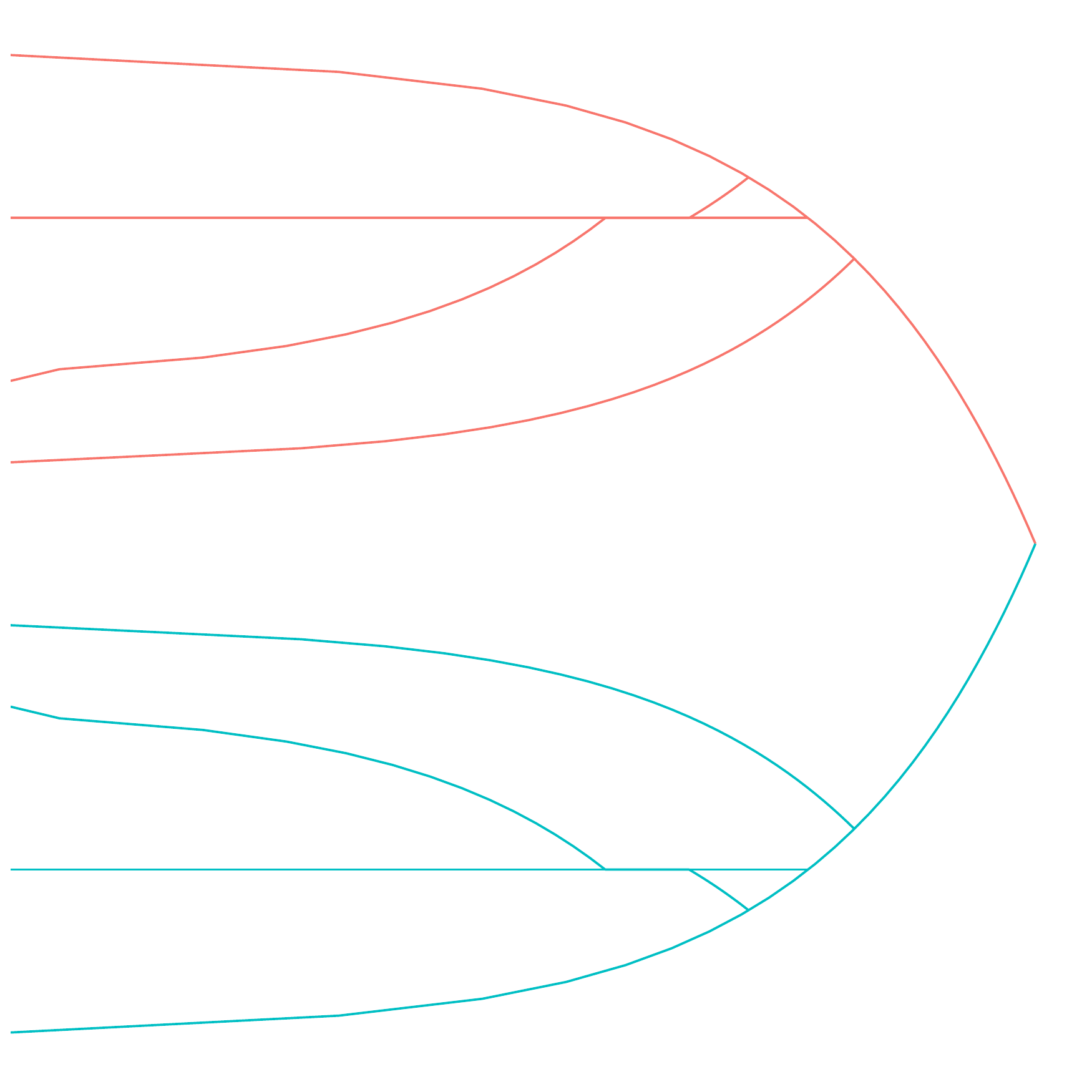} &
    \includegraphics[width=.3\textwidth]{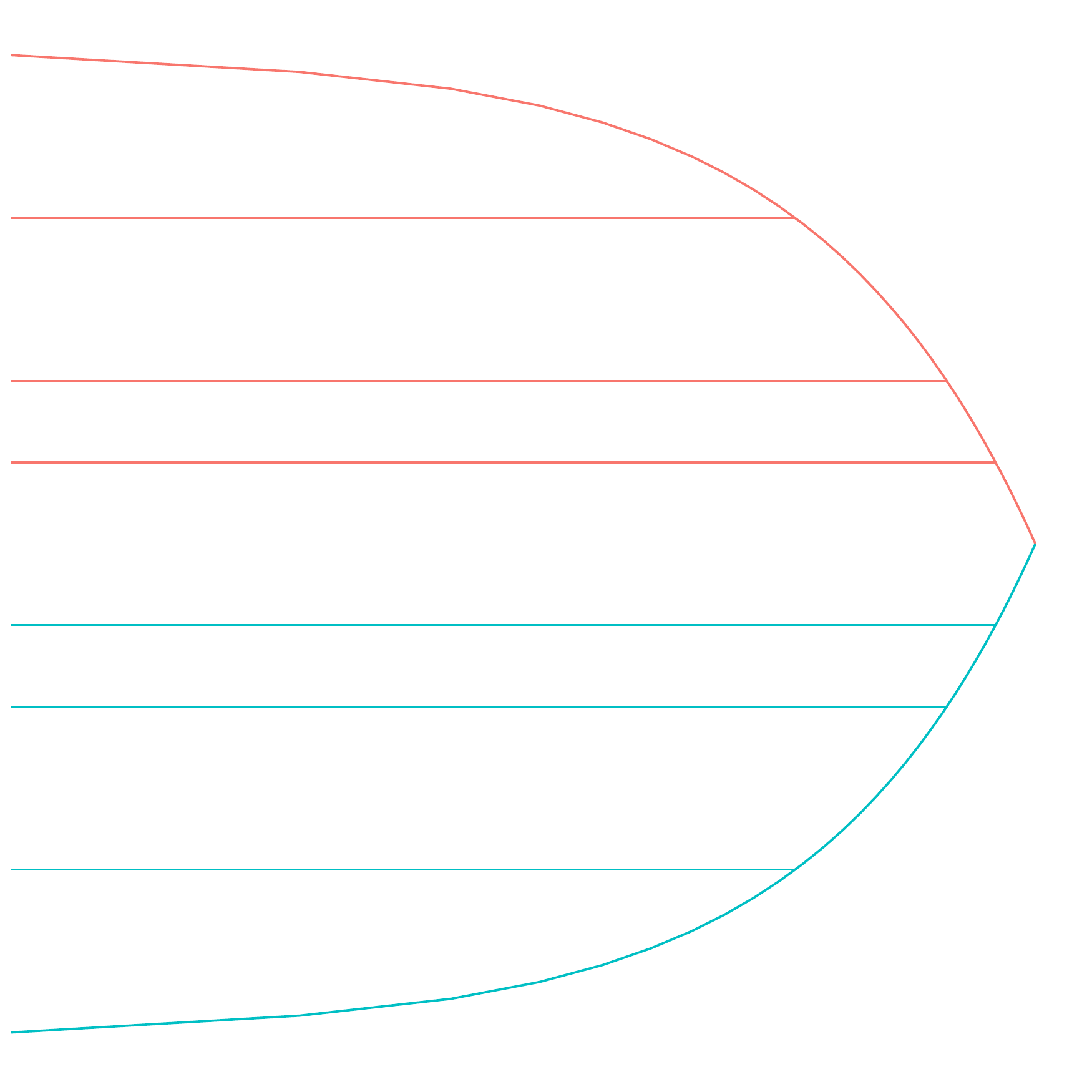} &
    \includegraphics[width=.3\textwidth]{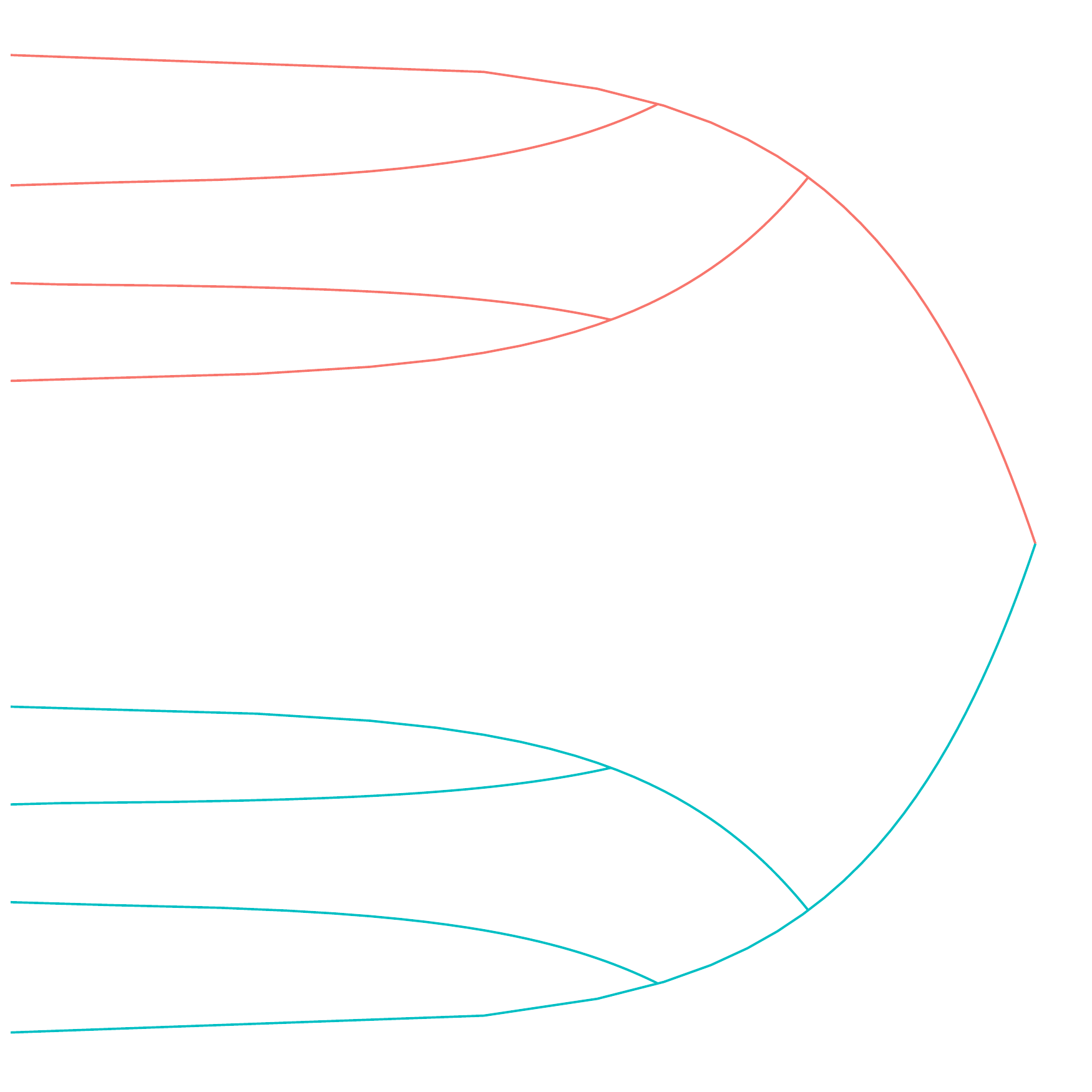} \\
    \rotatebox{90}{\hspace{1.1em}\small Piecewise linear path} &
    \includegraphics[width=.3\textwidth]{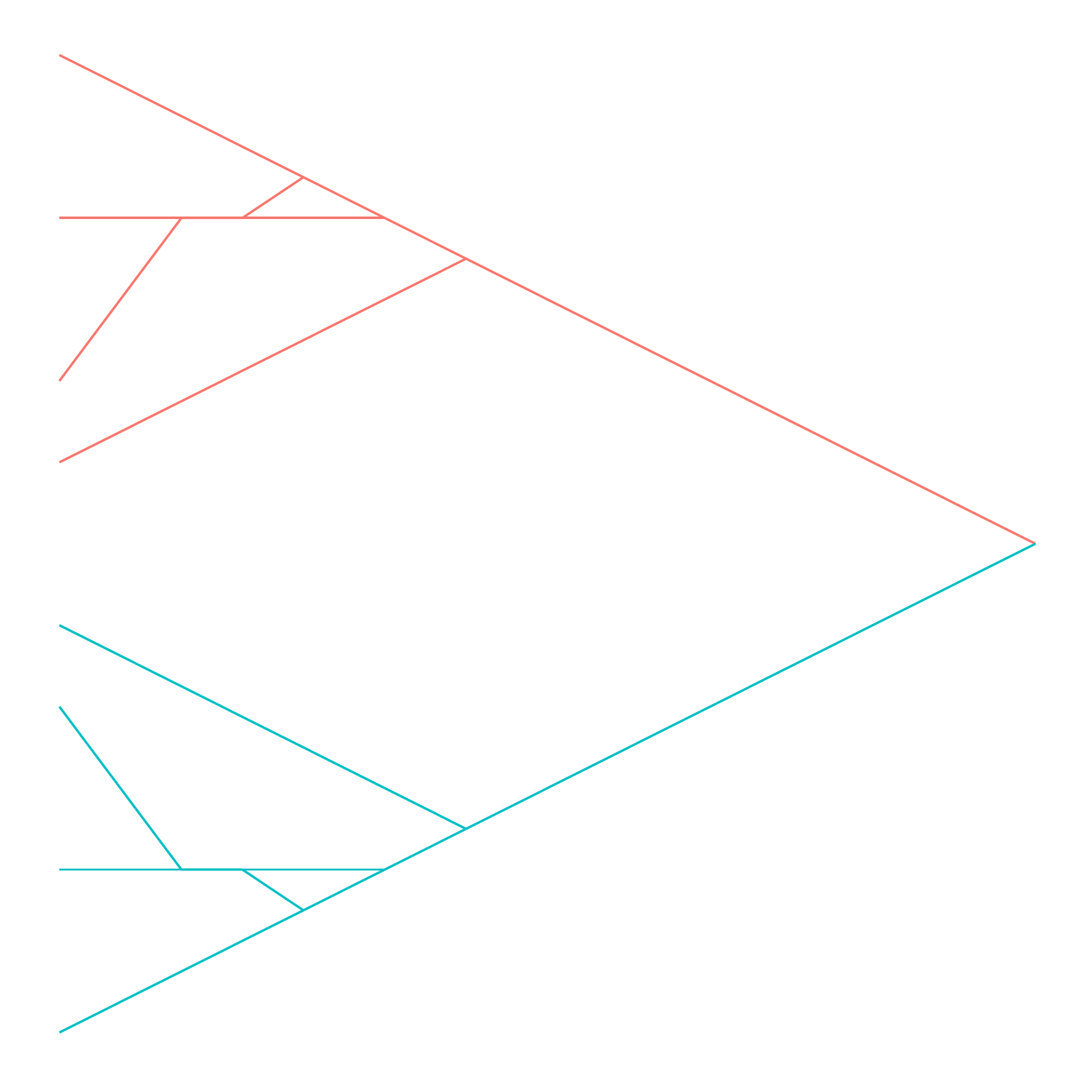} &
    \includegraphics[width=.3\textwidth]{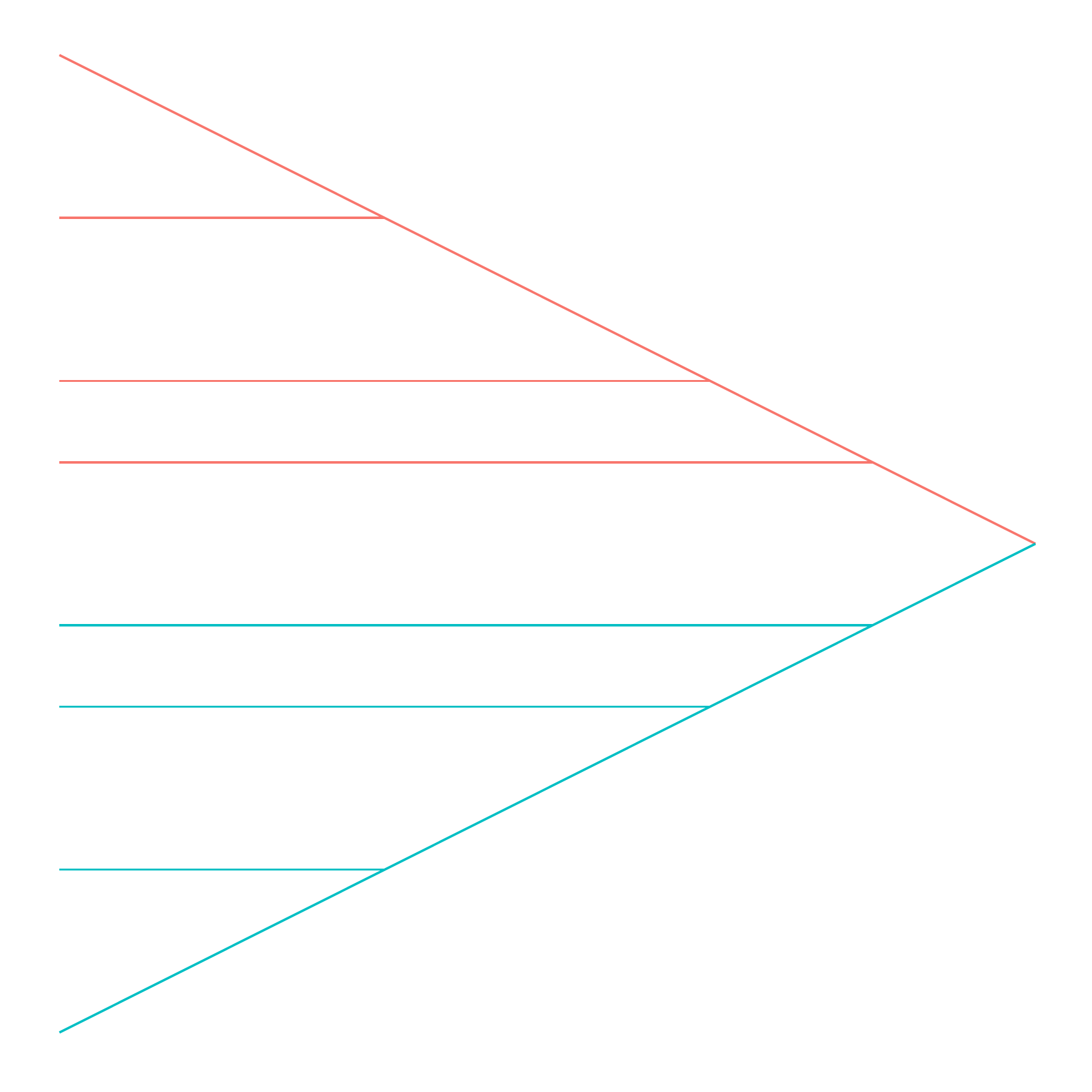} &
    \includegraphics[width=.3\textwidth]{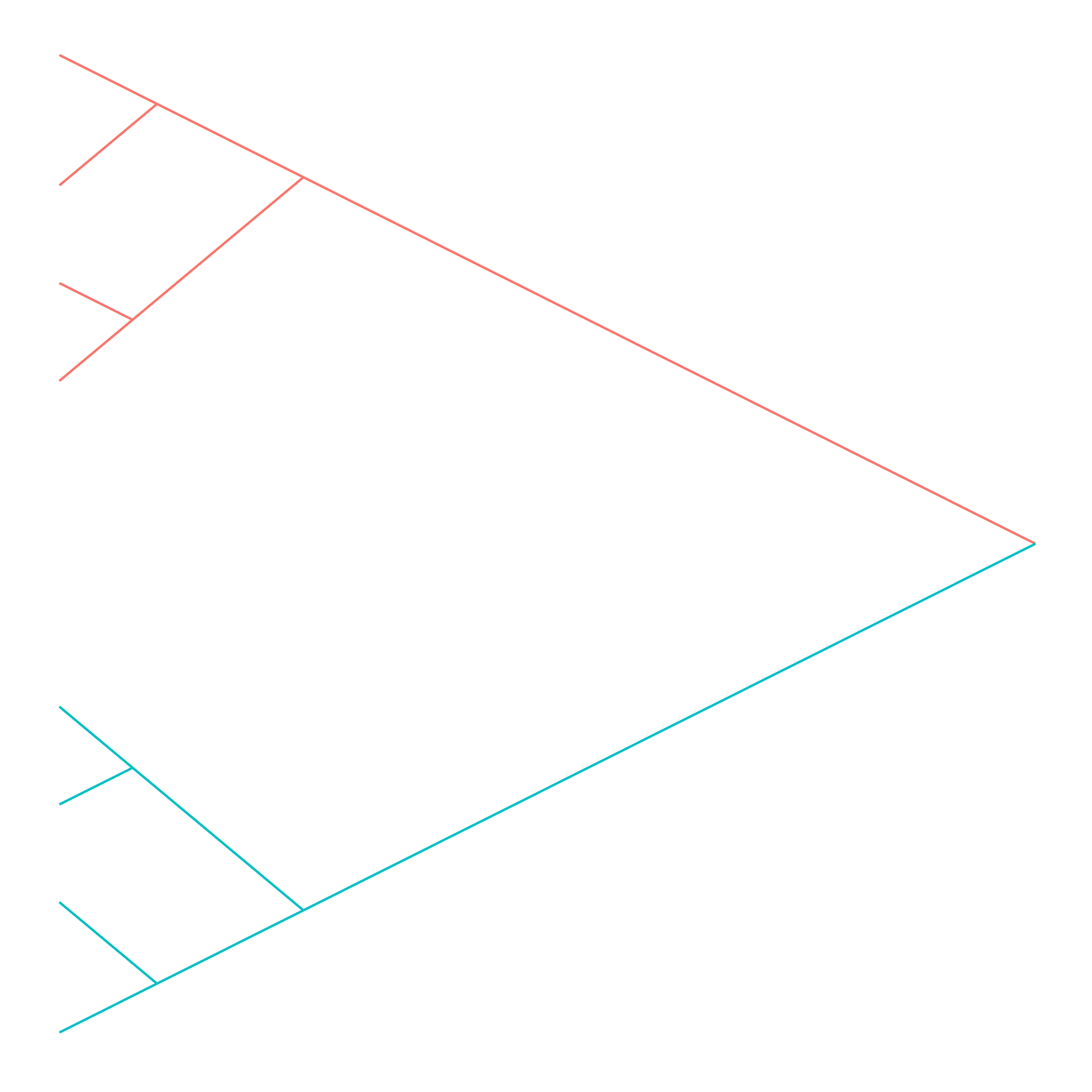} \\
    & \small not a tree (splits) &
    \small unbalanced tree & \small balanced tree \\
   \end{tabular}

   \caption{Various  typologies  of the  regularization  paths in  the
     single  feature case  that  lead to  more  or less  interpretable
     structures.}
  \label{fig:paths}
\end{figure}

Though highly  desirable, guaranteeing  a tree  is complicated  as the
absence of splits in  the path of \eqref{eq:criterion_general} depends
jointly on  the choice of  the weights  $w_{k\ell}$ and on  the fusing
norm $\Omega(\cdot)$.  In  the following Theorem, we  provide a simple
generic choice for  the weights that ensures the absence  of splits in
the general formulation with $\ell_q$-norms.
\begin{theorem}
  \label{thm:norm-nosplit}
  If $\Omega$ is an  $\ell_q$-norm with $q\in\set{1,\dots,\infty}$ and
  $w_{kl}   =    n_k   \cdot   n_\ell$,   the    path   of   solutions
  $\set{\bbeta(\lambda):\lambda>0}$   of  \eqref{eq:criterion_general}
  contains no splits.
\end{theorem}

The        proof        is         postponed        to        Appendix
\ref{sec:norm-nosplit}. Schematically, it investigates the subgradient
equations  of  \eqref{eq:criterion_general}  and shows  that  given  a
solution  at  $\lambda_0$,  we  can always  explicitly  construct  for
$\lambda  >\lambda_0$ a  valid  subgradient not  involving any  split.
Theorem    \ref{thm:norm-nosplit}   generalizes    the   results    of
\cite{HOCKING-clusterpath} obtained for $\Omega(\cdot) = \|\cdot \|_1$
in  the  clustering  case  when  $w_{k\ell}  =  n_k=n_\ell=1$  to  any
$\ell_q$-norm $\Omega$.

A consequence  of Theorem~\ref{thm:norm-nosplit} is that  the existing
implementation of  the $\ell_2$ Clusterpath  -- or any  other $\ell_q$
solver -- can  be simplified by no longer  considering the eventuality
of   splits    with   default    weights   \citep[see    Algorithm   1
in ][]{HOCKING-clusterpath}.
\\

As  said  before, guaranteeing  a  tree-structure  is the  first  step
towards  interpretability.   As such,  Theorem  \ref{thm:norm-nosplit}
characterizes an interesting family of problems.  Still, the scope of
arbitrary  norms  with  uniform  weights  is  not  fully  satisfactory
because, even when the structure is a tree,
\begin{itemize}
\item the  path is not a  linear function of $\lambda$  in general, as
  illustrated  on the first  row of  Figure \ref{fig:paths}.   In this
  situation, detecting the events of fusion may be expensive. It might
  be  impossible  to  provide  an  efficient algorithm  to  infer  the
  structure at a low computational cost.
\item the inferred structure may  be highly unbalanced. By unbalanced,
  we mean a  tree where two parameters initially close  to one another
  at $\lambda=0$ fuse relatively late  in the path of solutions.  Such
  situations   are   depicted  on   the   second   column  of   Figure
  \ref{fig:paths}. It is obvious that disequilibrium may significantly
  narrow the potential for interpretability of the tree.
\end{itemize}

First, equilibrium  of the  inferred structure is  a property  that is
mainly controlled  by the $w_{k\ell}$.   We cannot limit  ourselves to
$w_{k\ell} = n_k  \cdot n_\ell$ and must exhibit  weights sharing both
the equilibrium and  the non-split property. This will lead  us to the
distance-decreasing weights described in the next section.

Second,  piecewise-linearity   --  and   thus  existence  of   a  fast
path-following algorithm  -- is  a property of  the norm  $\Omega$.  A
solution path  which is piecewise  linear can be  computed efficiently
(and exactly)  with a homotopy algorithm  like the LARS for  the LASSO
\citep{efron2004least}.   More generally,  \citet{2007_AS_rosset} give
conditions for the  existence of such a property in  a broad penalized
framework.   These results  are easily  adapted to  the case  at hand,
where     we    roughly     have    to     differentiate    Expression
\eqref{eq:KKT_sum_beta}    over    $\lambda$    to    conclude:    for
$\Omega(\cdot)=\|\cdot\|_q$ any $q$-norm with $q\geq1$, then
  \begin{equation}
    \label{eq:beta_derivative}
    \frac{\partial  \bbeta_C}{\partial \lambda}  = \frac{1}{n_C}
    \sum_{\ell\notin  C} w_{C\ell}  \  \signs(\bbeta_\ell -  \bbeta_C)
    \circ \frac{\left| \bbeta_\ell - \bbeta_C \right|^{q-1}}{\left\| \bbeta_\ell -  \bbeta_C\right\|_q^{q-1}},
  \end{equation}
where $|\cdot|$  and $\signs(\cdot)$ apply  element-wise and $\circ$
is the element-wise product.
Application of  Proposition 1 of \citeauthor{2007_AS_rosset}  to these
expressions  implies  that  the  path is  piecewise  linear  only  for
$q\in\set{1,\infty}$.   In other  words, there  must exist  a homotopy
algorithm  to  infer the  structure  between  the conditions  for  the
$\ell_1$ and  $\ell_\infty$-norms.  More  generally, we could  use any
norm $\Omega$  that builds on  $\ell_1$ and $\ell_\infty$ such  as the
OSCAR \citep{Oscar2008}.   Note, however,  that there is  no guarantee
that the number  of steps will be small in  the homotopy algorithm for
general   weights.   In   fact,  \citet{ICML2012Mairal_202}   exhibits
pathological cases for the LARS algorithm where the number of kinks in
the piecewise  linear path of  solutions grows exponentially  with the
number of  variables.  Such  cases can be  transposed to  the weighted
fusion penalty  with $\Omega(\cdot) = \|\cdot\|_1$,  which corresponds
to situations where there is a  large number of splits along the path.
To  overcome  this  restriction  and  guarantee  that  the  number  of
iterations required to fit the whole  path of solutions will be small,
we introduce in  the next section a family of  weights that ensures no
split  along the  path of  solutions for  the particular  case of  the
$\ell_1$-norm.


\section{Distance-decreasing weights guaranteeing no split}
\label{sec:weights}

In this section, we focus  on the $\ell_1$-norm and generalize Theorem
\ref{thm:norm-nosplit}  to a  larger  class of  weights  that we  call
distance-decreasing       weights,       defined      in       Theorem
\ref{thm:weight-nosplit}. Indeed, although  uniform weights ensure the
absence of  split, the recovered  tree structure is  often unbalanced.
Intuitively,  distance-decreasing  weights  should ensure  that  close
neighbors fuse  quickly.  Here, we  demonstrate that for  such weights
there   is    no   split.     Thus,   the   algorithm    proposed   by
\cite{hoefling2010path} for the generalized fused-Lasso is 
considerably simplified since there is no need to check for possible split
events, and  thus there  is no need  to solve  potentially numerically
unstable maximum flow problems.

\begin{remark}
  Note that  the absence of splits  does not ensure a  fast algorithm.
  Indeed, the initialization of  the generalized fused-Lasso algorithm
  is for most weights in $K^2$.  We exhibit in Section \ref{sec:optim}
  a subset of distance-decreasing  weights for which initialization is
  linear and for which we can guarantee good statistical properties in
  Section \ref{sec:consistency}.
\end{remark}

Another advantage of the $\ell_1$-norm  is that it brings separability
across the  $p$ features  in \eqref{eq:criterion_general}, that  is to
say,  that  the $p$-dimensional  problem  splits  into $p$  univariate
problems. To  recover a consensus  classification, we first  infer $p$
independent trees  (one per  dimension) and  then aggregate  those $p$
trees by considering  the same penalty value  $\lambda$. Thus, without
loss  of  generality, we  restrict  the  discussion to  the  following
$\ell_1$   univariate  problem   which  is   a  weighted   generalized
fused-Lasso problem:
\begin{equation}
  \label{eq:criterion_univariate_l1}
  \minimize_{\bbeta \in \Rset^K} \frac{1}{2} \sum_{k=1}^K n_k \left( \bar{y}_k -
    \beta_k \right)^2
  + \lambda \ \sum_{k,\ell:k\neq\ell} w_{k\ell} |\beta_k - \beta_\ell|.
\end{equation}

For this problem, we get the following result:
\begin{theorem}\label{thm:weight-nosplit}
  The  path of  solutions does  not  contain splits  when weights  are
  chosen such that
  \begin{equation*}
    w_{k\ell} = n_k n_\ell \ f(\left|\bar{y}_k - \bar{y}_\ell \right|),
  \end{equation*}
  where $f(\cdot)$ is a decreasing positive function.
\end{theorem}

Schematically, the proof is based on two ingredients:
\begin{enumerate}
\item first, using geometrical arguments,  it is possible to show that
  absence of splits  is equivalent to preservation of  the order along
  the   path,   that  is   to   say,   $\bar{y}_k  \leq   \bar{y}_\ell
  \Leftrightarrow              \hat{\beta}_k(\lambda)             \leq
  \hat{\beta}_\ell(\lambda)$;
\item   second,   by  considering   a   problem   that   is  dual   to
  \eqref{eq:criterion_univariate_l1}  as  in \cite{2011_AS_Tibshirani}
  for the generalized Lasso,  we show that distance-decreasing weights
  preserve the order.
\end{enumerate}

The proof is detailed in Appendix \ref{sec:nosplit}.


\section{Fast homotopy algorithm for $\ell_1$ weighted penalties}
\label{sec:optim}

In this section,  we consider algorithmic issues when  $\Omega$ is the
$\ell_1$-norm.   As  in  Section~\ref{sec:weights},  we  restrict  the
discussion  to  univariate Problem  \eqref{eq:criterion_univariate_l1}
and  thus give  the numerical  complexity in  the case  $p=1$.  For  a
$p$-dimensional  problem, we  aggregate  the $p$  univariate trees  by
considering the same values of $\lambda$ for all trees.

\paragraph*{An algorithm for general weights and its limitations.}
Optimization problem  \eqref{eq:criterion_univariate_l1} can be solved
for general weights $w_{k\ell}$  by the homotopy algorithm proposed in
\cite{hoefling2010path} for the generalized fused-Lasso.  This is also
the    solution   retained    in   the    clustering    framework   by
\cite{HOCKING-clusterpath}.   A  schematic   view  of  this  algorithm
adapted   to   \eqref{eq:criterion_univariate_l1}   is   depicted   in
\ref{algo:homotopy}.

\begin{algorithm}[H]
  \DontPrintSemicolon
  \BlankLine
  \KwIn{data, weights and initial groups $\set{y_i,w_{k\ell},\kappa}$}
  \textbf{Initialization for $\lambda=0$}\;
  Initialize $\beta_k$ parameters (equal to the empirical means $\bar{y}_k$)\;
  Initialize the list of possible next events (only fusion at this stage) \;

  \While{all groups are not fused} {
    Find the next  event (having the smallest $\lambda$),  it can be a
    split or a fusion \;
    Update $\beta_k$ parameters accordingly \;
    Update the list of possible next events (fusion and split)
  }
  \KwOut{Directed acyclic graph (DAG) of fusion and split events and associated values of the
    parameters}

  \caption{Homotopy      algorithm       for      the      generalized
    fused-Lasso}\label{algo:homotopy}
\end{algorithm}

This procedure for  general weights has two major  flaws that may have
detrimental effects on its computational performance:
\begin{itemize}
\item By piecewise-linearity of the solution path, the total number of
  iterations  (that is,  the total  number  of events  before all  the
  groups   have   fused)   is    bounded.    However,   by   rewriting
  \eqref{eq:criterion_univariate_l1} as a Lasso  problem -- which only
  requires straightforward  algebra --  we may  construct pathological
  cases where there  are $(3^K + 1)/2$ linear segments  in the path of
  solutions  \citep[see][]{ICML2012Mairal_202}, a  complexity that  we
  cannot afford even for a moderate number of conditions $K$.
\item While  detecting fusion events in  Algorithm \ref{algo:homotopy}
  may  be cheap  since it  roughly  only requires  calculation of  the
  slopes  $\partial \beta_k(\lambda)/\partial  \lambda$, checking  for
  the possibility of split events  boils down to maximum-flow problems
  the resolution of  which at large scale may clearly  be a bottleneck
  \citep[see][]{hoefling2010path}.
\end{itemize}

To  circumvent  these  limitations,  we shall  consider  weights  that
prevent split  events.  Although the  choice $w_{k\ell} =  n_k n_\ell$
has been shown to prevent splits in Theorem \ref{thm:norm-nosplit}, it
will typically lead to fusion events occurring very late (that is, for
large $\lambda$),  even between  groups having close  empirical means.
This  corresponds   to  an  unbalanced  tree   structure  between  the
conditions, which is hardly interpretable.  On the contrary, using the
family   of  distance-decreasing   weights,   introduced  in   Section
\ref{sec:weights}, prevents split events and  leads to a balanced tree
structure.  In this case the total  number of events is exactly $K-1$,
which is the number of iterations  required to fuse $K$ groups into 1,
assuming that  there cannot  be a  fusion of more  than two  groups at
once. As  for the  maximum-flow problems,  they are  completely eluded
from the  algorithm with these  weights. Still,  we have to  take into
account the cost of detecting successive fusion events and of updating
the coefficients $\beta_k(\lambda)$ along the $K-1$ steps. In the next
paragraph,  we propose  a  solution inducing  a  global complexity  of
$\mathcal{O}(K\log(K))$ for a given choice of weights belonging to the
family of distance-decreasing weights.

\paragraph*{Weights with an $\mathcal{O}(K\log(K))$ implementation.}
First we  need to  define the  next time  a fusion  event is  going to
happen.   We  proceed mainly  as  in  \cite{hoefling2010path} for  the
one-dimensional  fused-Lasso  signal  approximator,  except  that  the
initial  ordering  is not  defined  by  the neighborhood  between  the
coefficients, but by the ordering  of the empirical means $\bar{y}_k$.
And thanks  to the property  of the distance-decreasing  weights, this
ordering remains the same throughout the algorithm, which allows us to
compute the path in $\mathcal{O}(K\log  K)$ operations.  Here are some
details.

At the initialization step, one has $\lambda_0=0$, and the next time a
fusion occurs is
\begin{equation}
  \label{eq:fusion_time}
  t(\lambda) = \argmin_{t_{k\ell}(\lambda) > \lambda_0} t_{k\ell} , \quad
  t_{k\ell}(\lambda) = \lambda_0 - (\beta_k(\lambda_0) - \beta_\ell(\lambda_0))
  \left(\frac{\partial \beta_k}{\partial \lambda}(\lambda_0)
    -   \frac{\partial  \beta_\ell}{\partial\lambda}(\lambda_0)   \right)^{-1}.
\end{equation}
In words, it  is the smallest value of $\lambda$  among all the values
such    that   two    coefficients    fuse.   The    main   cost    in
\eqref{eq:fusion_time}  is  due  to  the  calculation  of  the  slopes
$\partial  \beta_k/\partial \lambda$  at $\lambda_0  = 0$.   Note that
$\beta_k(0)  =  \bar{y}_k$, and  by  Lemma \ref{lem:KKT_sum_beta}  and
\eqref{eq:beta_derivative}, one has
\begin{equation}
  \label{eq:slopes}
  \frac{\partial \beta_k}{\partial \lambda}(0) = - \frac{1}{n_k}
  \sum_{\ell\neq k} w_{k\ell} \ \signs(\bar{y}_k - \bar{y}_\ell).
\end{equation}
For general  weights $w_{k\ell}$, computing  these slopes for  all $k$
requires $\mathcal{O}(K^2)$  operations and is the  limiting factor of
the  algorithm.   However,  we  provide a  $\mathcal{O}(K  \log  (K))$
procedure for a  special case of our  distance-decreasing weights that
we   call  ``exponentially   adaptive  weights''   because  of   their
statistical properties (see  Section \ref{sec:consistency}).  They are
defined by
\begin{equation}
  \label{eq:fa_weights}
  w_{k\ell}  =  n_k  n_\ell  \exp\{-\alpha  \sqrt{n}
  |\bar{y}_k - \bar{y}_\ell|\}, \quad \alpha > 0,
\end{equation}
for  $\alpha$   a  positive  constant.    The  key  idea   to  achieve
$\mathcal{O}(K\log(K))$  complexity with  these weights  is that  each
slope can be computed as the sum  of two terms, for which there exists
a  simple recurrence  formula: first,  we order  the $\bar{y}_{k}$  in
decreasing  order, which  can  be done  in  $\mathcal{O}(K \log  (K))$
operations. Assuming this is done, we obtain


\begin{align*}
  \frac{\partial \beta_k}{\partial \lambda}(0) & = -
  \sum_{\ell\neq k} n_\ell \ \signs(\bar{y}_k - \bar{y}_\ell) \ \exp\set{-\alpha \sqrt{n}|\bar{y}_k - \bar{y}_\ell|} \\
  &  =  \sum_{\ell  <  k  }  n_\ell  \exp\set{-\alpha \sqrt{n}(\bar{y}_\ell  -    \bar{y}_k)}
  - \sum_{\ell > k } n_\ell \exp\set{-\alpha \sqrt{n}(\bar{y}_k  -    \bar{y}_\ell)}\\
  & =   \exp\set{\alpha  \sqrt{n}\bar{y}_k} \underbrace{\sum_{\ell <  k }
    n_\ell \ \exp\set{-\alpha \sqrt{n}\bar{y}_\ell}}_{L_k}
  - \exp\set{-\alpha \sqrt{n}\bar{y}_k}\underbrace{\sum_{\ell > k } n_\ell  \ \exp\set{\alpha \sqrt{n}\bar{y}_\ell}}_{R_k}.
\end{align*}
The recurrence  formulae are $R_{k+1}  = R_{k} +  n_k \exp\set{-\alpha
  \bar{y}_k}$ and $L_{k-1} = L_k + n_k \exp\set{\alpha \bar{y}_k}$. By
this means,  the initial slopes  \eqref{eq:slopes} and thus  the first
fusion time can be computed in $\mathcal{O}(K\log(K))$.

Then, for each  of the $K-1$ steps  of the algorithm, we  only need to
update the  two slopes  and the two  coefficients which  are currently
fusing.   This   only  requires  a  constant   number  of  operations.
Concerning the  next fusion time,  however, the new minimum  among the
updated $t_{k\ell}(\lambda_0^+)$ is found in $\log(K)$ if stored in an
appropriate structure.  This way  we can reach $\mathcal{O}(K\log(K))$
for the global complexity.

As  a final  remark, note  that we  use the  same storage  solution --
namely  a  binary  tree  -- as  did  \cite{hoefling2010path}  for  the
one-dimensional fused-Lasso.   By this  means, we maintain  the memory
requirement at a low level that only grows linearly in $K$.

\paragraph*{An  embedded cross-validation  procedure.}  Providing  the
whole path  of solutions is clearly  interesting for interpretability,
since we force  it to be a  tree by means of  an appropriate weighting
scheme coupled with the $\ell_1$-norm  for fusion. Still, it is always
necessary to provide  a practical way to choose  the tuning parameter,
which corresponds in  the case at hand to choosing  the level at which
to cut the  tree.  This also gives a fixed  classification between the
initial conditions.

When   the  number   $K$  of   prior  groups   is  smaller   than  $n$
(\textit{e.g.},  in the  ANOVA settings),  a natural  cross-validation
(CV) error can be defined. Although CV is often incriminated for being
time-consuming,  it is  possible  in this  case to  rely  on the  tree
structure of the solution -- or DAG in the case where split is allowed
in the algorithm -- to enhance  the performance.  Indeed, we can first
build  a tree  using a  training set  (in which  all prior  groups are
present) and then  assess its performance by measuring  its ability to
predict the remaining individuals of the  test set for any given value
of $\lambda$. Here,   we perform  the CV on  a predefined grid  of $L$  values of
$\lambda$ because  the fusion  times will be  different for  every new
training set  and it would be  memory intensive to store  the CV-error
for every of those fusion time.

To be more  specific, we consider a  split of the data in  a train set
$\mathcal{D}$ and a test set  $\mathcal{T}$ such that each prior group
is represented  in the train  set.  Using $\mathcal{D}$, we  recover a
fused-ANOVA tree  and an  estimator $\hatbbetaD_{\kappa(i)}(\lambda)$.
The test error on $\mathcal{T}$ is
\begin{equation}
  \label{eq:cv_err}
  \mathrm{CV}_{\mathrm{err}}(\mathcal{D}, \mathcal{T}, \lambda) = \sum_{i \in {\mathcal{T}}} \left(y_i - \hatbbetaD_{\kappa(i)}(\lambda)\right)^2 .
\end{equation}
A naive  approach to computing  \eqref{eq:cv_err} is to  consider each
prior  group  at  a  time  on a  grid  of  $\lambda$.   Computing  the
prediction based on a \emph{given} fitted regularization path requires
$\mathcal{O}(\log(K))$  operations  to  search  through  the  tree  of
solutions.  This has to  be done for the $K$ prior  groups and for the
$L$  values  in the  grid  of  $\lambda$.   Hence, computing  the  sum
\eqref{eq:cv_err}  naively has  a total  complexity of  $\mathcal{O}(L
K\log(K))$ (which dominates  the complexity in $\mathcal{O}(K\log(K))$
of the fit itself!).

On  the  contrary,  our   embedded  cross-validation  procedure  takes
advantage  of  the tree  structure  of  the  fit in  the  computations
whenever possible.   Indeed, along  the branch of  a cluster  $C$, the
estimator  $\hatbbetaD_{\kappa(i)}(\lambda)$  is  a  piecewise  linear
function  of  $\lambda$  and  thus the  error~\eqref{eq:cv_err}  is  a
piecewise quadratic  function of $\lambda$.  The  coefficients of this
quadratic function are easily updated  when constructing the tree, and
the error  along this branch  is computed in $\mathcal{O}(1)$  for any
$\lambda$ rather  than $\mathcal{O}(|C|\log(K))$.  More  precisely the
error  in~\eqref{eq:cv_err} of  cluster $C$  decomposes thanks  to the
Huygens formula as
\begin{displaymath}
  \sum_{i   \in  \mathcal{T}   :  \kappa(i)\in   C  }   \left(y_i  -
    \hatbbetaD_{\kappa(i)}(\lambda)\right)^2 = 
  \sum_{i   \in  \mathcal{T}   :  \kappa(i)\in   C} (y_i - \bar{\by}_C^{\mathcal{T}} )^2 + n_C^{\mathcal{T}} \left(\bar{\by}_C^{\mathcal{T}} - \hatbbetaD_{\kappa(i)}(\lambda)\right)^2,
\end{displaymath}
where    $n_C^{\mathcal{T}}=\mathrm{card}(\set{i   \in    \mathcal{T}:
  \kappa(i)\in C})$  and $\bar{\by}_C^{\mathcal{T}}$ is  the empirical
mean of individuals of cluster $C$, \textit{i.e.},
\begin{displaymath}
  \bar{\by}_C^{\mathcal{T}}  =
   \frac{1}{n_C^{\mathcal{T}}} \sum_{i \in \mathcal{T} : \kappa(i)\in C} y_i.
\end{displaymath}

It is difficult  to assess exactly the gain brought  by using the tree
structure for computing  the CV error in general.   Indeed, it depends
on the  tree itself,  the length  of its branches,  its height  and so
on. Assuming a binary balanced tree of height $\log(K)$, with branches
of equal length  and an equally spaced grid of  $\lambda$, we can show
that the complexity is in $\mathcal{O}(LK/ \log (K))$.  If some groups
fused rapidly  (as with  the fused-ANOVA weights),  the gain  could be
even greater.  In practice  (see Figure~\ref{fig:timings}.c), we often
see  a  ten-fold difference  between  our  CV  procedure and  a  naive
implementation.

\paragraph*{Timings.}   We  implemented  both   the  general  and  the
without-split version of Algorithm \ref{algo:homotopy} in \texttt{C++}
embedded   in   an   \texttt{R}-package   called   \texttt{fusedanova}
distributed on \texttt{R}-forge. It contains  a wide family of weights
which  are not  mentioned in  this  paper due  to space  requirements.
Figure \ref{fig:timings}  illustrates the  rather good  performance of
our algorithm and implementation through three numerical experiments:
\begin{enumerate}[$a)$]
\item In the left panel, we illustrate the capability of our method to
  treat large  scale problems extremely  fast: we generate  a size-$n$
  vector $\by$  such that $y_i\sim\mathcal{N}(0,1)$ and  assume $n=K$,
  meaning  one  condition   per  group\footnote{With  this  simulation
    setting, there is  no underlying clustering since our  point is to
    compare run  times here.}. We vary  $n$ from $10^2$ to  $10^8$ and
  record the  corresponding timing  in seconds.   We apply  our method
  with the exponentially adaptive weights  and average over 10 trials.
  As can be  seen, we can reconstruct a tree  on $n=10^6$ observations
  in about 10 seconds.
\item The middle panel illustrates the gain in runtime due to the fact
  that we no longer have to check for splits in the homotopy algorithm
  using a maximum-flow  solver.  We generate data as  in the preceding
  experiment  but   with  $K$  conditions  each   containing  $n_k=20$
  replicates.   When $K=10^3$,  the  gain in  seconds  brought by  not
  checking for  the possibility  of splits  is of  almost 2  orders of
  magnitude.
\item The right  panel illustrates the performance of  our embedded CV
  procedure compared  to the naive  implementation.  We used  the same
  settings as in the previous experiment.
\end{enumerate}

\begin{figure}[htbp!]
  \centering
  \begin{tabular}{@{}l@{}c@{}c@{}c@{}}
    \rotatebox{90}{\hspace{1.75em}\footnotesize timings in seconds (log)} &
    \includegraphics[width=.32\textwidth]{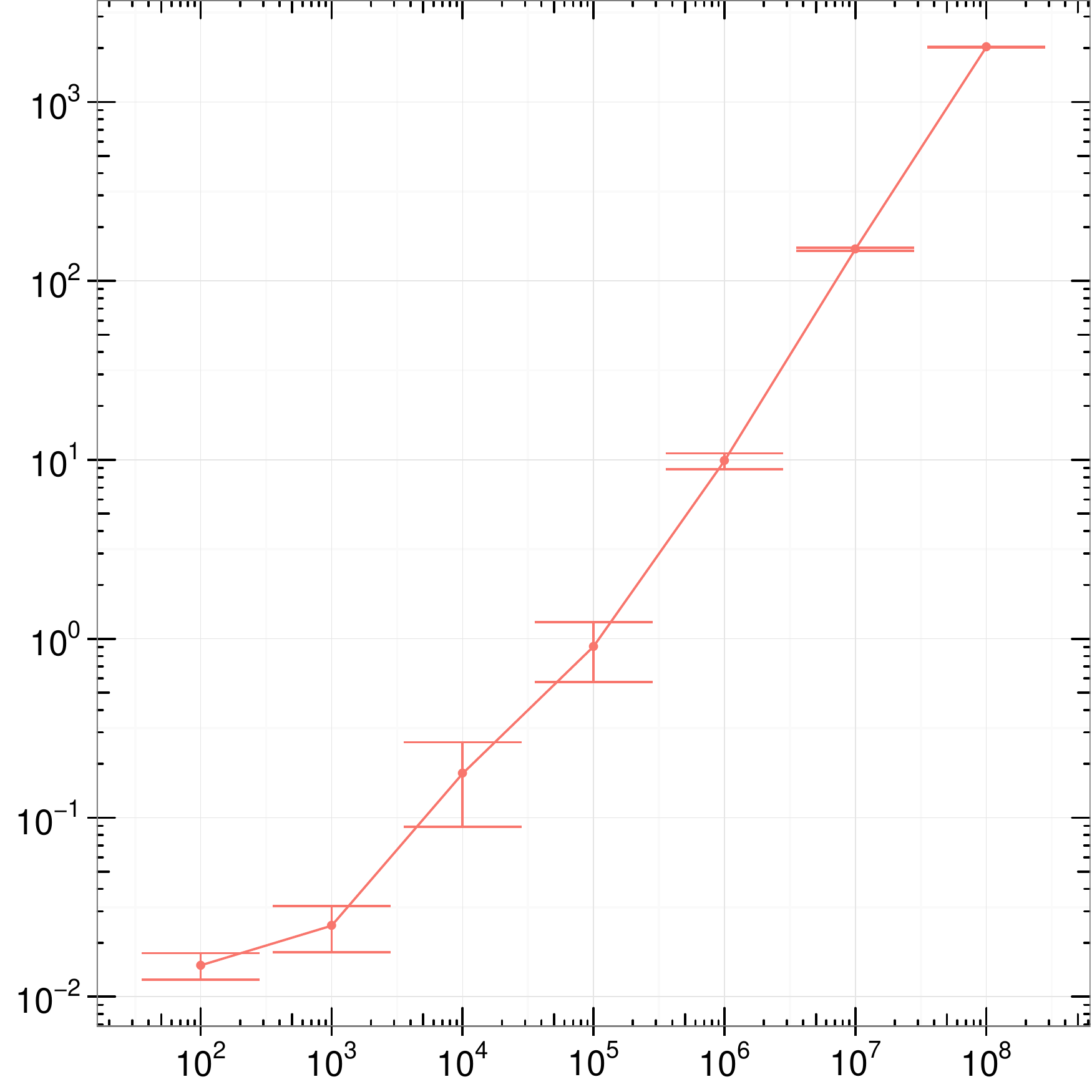} &
    \includegraphics[width=.32\textwidth]{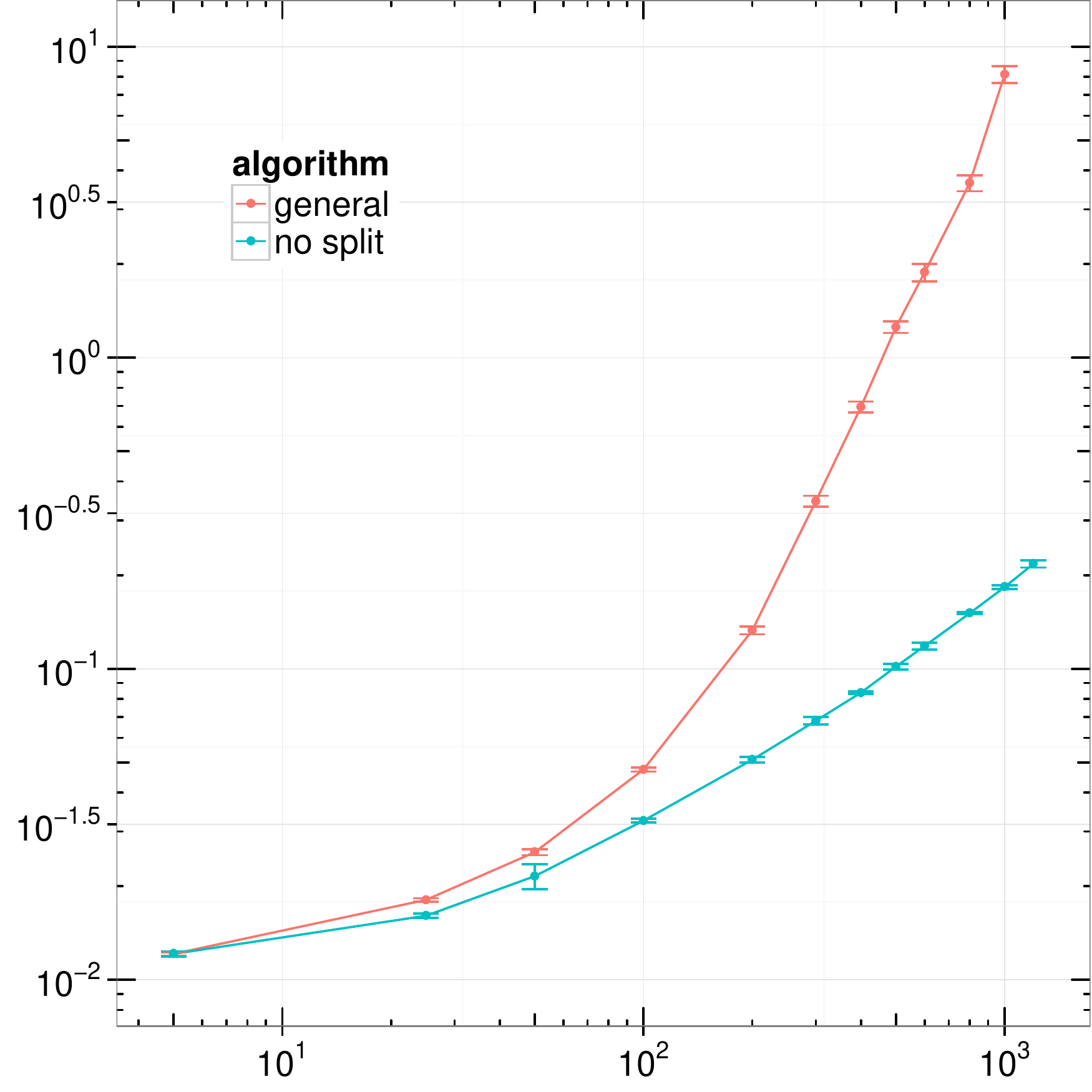} &
    \includegraphics[width=.32\textwidth]{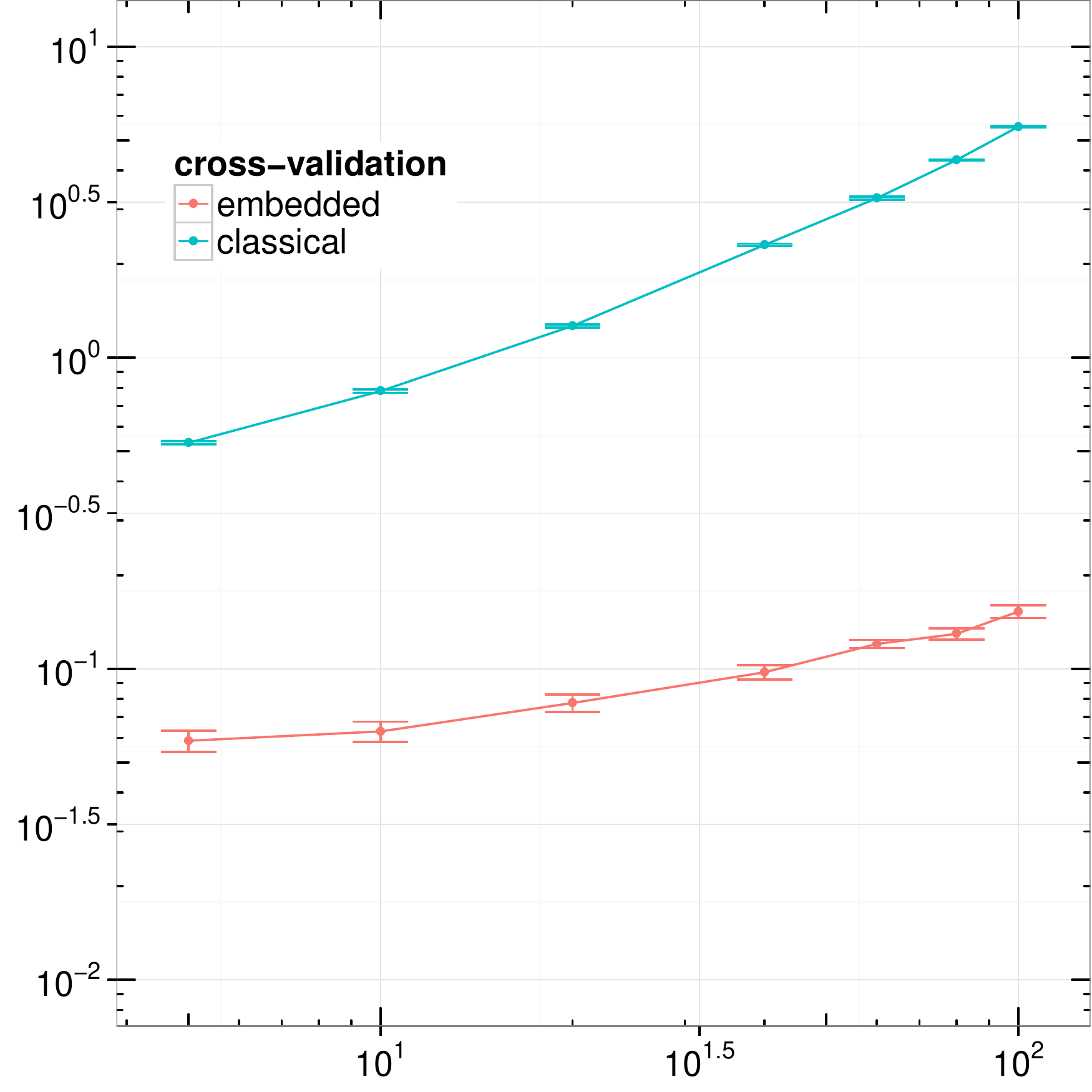} \\
    & \multicolumn{3}{c}{\footnotesize number of conditions $K$ (log)}
  \end{tabular}
  \caption{timing experiments: $a)$  time in seconds as  a function of
    the  number   of  conditions  $K$;  $b)$   timing  comparison  for
    general/without-split  algorithm; and  $c)$ timing  comparison for
    naive/embedded cross-validation.}
  \label{fig:timings}
\end{figure}

We       tried       other        implementations       to       solve
\eqref{eq:criterion_univariate_l1}  such  as the  \texttt{Clusterpath}
package  by \cite{HOCKING-clusterpath},  the \texttt{flsa}  package by
\cite{hoefling2010path}   or   the    \texttt{genlasso}   package   by
\cite{2011_AS_Tibshirani}. These implementations  do not fully exploit
the structure  of the  problem and  have runtimes  considerably longer
than ours, even for moderate $K$. Thus, we do not report their timings
here.


\section{Statistical guarantees}
\label{sec:consistency}

\paragraph*{Asymptotic settings.} To discuss the asymptotic properties
of our exponentially adaptive  weights \eqref{eq:fa_weights}, we shall
consider  the following  univariate\footnote{We numerically  study the
  multidimensional case at the end of this section.} ANOVA
model
 \begin{equation}
   \label{eq:model_fa}
   y_i = \beta_{\kappa(i)}^\star + \varepsilon_i, \quad \text{s.t.} \quad
   \E(\varepsilon_i)=0, \ \var(\varepsilon_i)=\sigma^2, \quad i=1,\dots,n,
 \end{equation}
 where $\bbeta^\star=(\beta_1^\star,\dots,\beta_K^\star)$  is the true
 vector  of parameters  and  $\varepsilon_i$ are  iid residuals.   The
 correct structure between the coefficients -- or classification -- in
 $\bbeta^\star$  is denoted  by  $\mathcal{A}^{\star}  = \set{(k,l)  :
   \beta_k^\star =  \beta_\ell^\star}$.  A usual  technical assumption
 is to consider designs the associated gram matrices of which converge
 to positive  definite matrices.   In the  one-way ANOVA  settings, we
 just need to assume that when  $n\to\infty$, then $n_k/n \to \rho_k <
 \infty$ for all $k=1,\dots,K$.  We  denote by $\bD$ the corresponding
 asymptotic  covariance matrix  which  is a  $K$-diagonal matrix  with
 diagonal entries equal to $\rho_1,\dots,\rho_K$.

 In  the univariate  case like  in \eqref{eq:model_fa},  the estimator
 associated  with   Problem  \eqref{eq:criterion_general}   using  the
 $\ell_1$-norm for fusion is
\begin{equation}
  \label{eq:fused_anova}
  \hatbbeta^{(n)} = \argmin_{\bbeta \in \Rset^K} \frac{1}{2} \sum_{k=1}^K
  n_k \left(\bar{y}_k - \beta_k \right)^2
  +  \lambda_n  \  \sum_{k\neq\ell}  w_{k\ell}  |  \beta_k  -\beta_\ell |,
\end{equation}
which is just  a rewriting of \eqref{eq:criterion_univariate_l1} where
the dependency  on $n$  of the estimator  and the tuning  parameter is
stated explicitly for the  purpose of asymptotic analysis.  Similarly,
we denote  by $\hat{\mathcal{A}}_n = \set{(k,\ell):\hat{\beta}_k^{(n)}
  = \hat{\beta}_\ell^{(n)}}$ the estimated group structure.

\paragraph*{Exponentially  adaptive weights  and the  fused-ANOVA.} In
this paragraph, we study the  exponentially adaptive weights, which we
recall here:
\begin{equation*}
  w_{k\ell}^{\text{FA}}  =  n_k  n_\ell  \exp\{-\alpha  \sqrt{n}
  |\bar{y}_k - \bar{y}_\ell|^{\gamma}\}, \quad \alpha, \gamma > 0.
\end{equation*}
We show  that they enjoy  some ``oracle  properties'' in the  sense of
\cite{2001_JASA_fan}, that  is, both  $i)$ right  model identification
(recovering  the true  classification  $\mathcal{A}^\star$) and  $ii)$
optimal estimation rate  $\sqrt{n}$.  In the context  of the penalized
ANOVA  problem  \eqref{eq:fused_anova},  we denote  these  weights  by
$w_{k\ell}^{\text{FA}}$   and  call   the  associated   estimator  the
\emph{fused-ANOVA}.    These   weights   are  adaptive   as   in   the
adaptive-Lasso  of  \cite{zou2006adaptive}:  it   is  known  that  raw
$\ell_1$ methods like the Lasso do not enjoy the aforementioned oracle
properties, yet this  can be fixed by choosing  judicious weights that
depend  on  an estimator  of  $\bbeta^\star$  which is  asymptotically
$\sqrt{n}$-consistent -- like the ordinary least squares, which equals
$(\bar{y}_1,  \dots, \bar{y}_K)$  in the  case at  hand.  Here  we are
interested in the differences between the entries of $\hatbbeta$; thus
the quantity $\sqrt{n}|\bar{y}_k -  \bar{y}_\ell|$ seems quite natural
in \eqref{eq:fa_weights}.

While studying  the asymptotic  of our estimator,  we came  across the
proposal of \cite{bondell2008simultaneous}  for adaptive weights: they
consider Problem \eqref{eq:fused_anova} with additional constraints on
the $\beta_k$'s -- that must sum to zero -- and the following weights,
which we refer to as the \emph{Cas-ANOVA} weights:
\begin{equation}
  \label{eq:ca_weights}
  w_{k\ell}^{\text{CA}} = \frac{\sqrt{n_k + n_\ell}}{|\bar{y}_k - \bar{y}_\ell|}.
\end{equation}
As  we shall  see,  though quite  interesting,  Cas-ANOVA weights  are
adaptive  on  a smaller  range  of  $\lambda_n$ than  are  fused-ANOVA
weights.   Moreover,  they lead  to  splits.   Thus, we  believe  that
fused-ANOVA is  computationally and  statistically more  efficient for
solving Problem \eqref{eq:fused_anova}.

We  now proceed  to the  Theorem  stating the  required conditions  on
$\lambda_n$ for the fused-ANOVA to enjoy the oracle properties.
\begin{theorem}[Oracle  properties]   \label{thm:oracle_prop}  Suppose
  that  $\lambda_n   n^{3/2}  \exp\set{-\alpha\sqrt{n}}  \to   0$  and
  $\lambda_n  n^{3/2}   \to  \infty$  when  $n\to\infty$.    Then  the
  fused-ANOVA   enjoys  asymptotic   normality  and   consistency  for
  recovering the true classification, \textit{i.e.},
  \begin{equation*}
    \sqrt{n}\left(\hat{\bbeta}^{(n)}   -   \bbeta^\star\right)\  \to_d
    \mathcal{N}(\bzr,
    \sigma^2 \bD^{-1})
    \quad \text{and} \quad
    \mathbb{P}(\hat{\mathcal{A}}_n   =    \mathcal{A}^\star)   \to   1
    \text{ when } n\to\infty.
  \end{equation*}
\end{theorem}

The proof  is postponed to Appendix  \ref{sec:oracle_prop} and roughly
follows that of \citeauthor{zou2006adaptive}.   We have, however, some
comments related to this Theorem.

\begin{remark}[On the exponentially adaptive weights]
  The key idea behind this theorem  is that when $n$ goes to infinity,
  then $ w_{k\ell}^{\text{FA}}/\sqrt{n}$ goes to infinity if $(k,\ell)
  \in \mathcal{A}^\star$  and to zero exponentially  fast if $(k,\ell)
  \notin \mathcal{A}^\star$.  This  is due to the joint  effect of the
  $\sqrt{n}$-consistency  of the $\bar{y}_k$  and of  the exponential.
  This  is  to  be   compared  with  Cas-ANOVA  weights,  where,  when
  $n\to\infty$,  $w_{k\ell}^{\text{CA}}/\sqrt{n}$ goes to  infinity if
  $(k,\ell)  \in  \mathcal{A}^\star$,  but   only  to  a  constant  if
  $(k,\ell) \notin \mathcal{A}^\star$.
\end{remark}

\begin{remark}[On the range of $\lambda_n$]
  Theorem  \ref{thm:oracle_prop}   is  true  for  a   large  range  of
  $\lambda_n$  values.   In  particular  it is  true  for  a  constant
  $\lambda_n$.  Asymptotically all groups  belonging to the same class
  fuse  almost   immediately  (\textit{i.e.},  for  small   values  of
  $\lambda$ of the  order $ n^{3/2} \exp  \set{-\alpha \sqrt{n}}$) and
  the  groups  belonging to  different  classes  fuse for  very  large
  $\lambda$, \textit{i.e.}, of the order $n^{3/2}$.
\end{remark}

\paragraph*{Numerical  illustration  in   the  univariate  case.}   We
generate data from  model \eqref{eq:model_fa} as follows,  for $K$ the
number  of prior  groups  and  $n$ being  fixed:  the  true vector  of
parameters  $\bbeta^\star$  is  composed  of  $K$  entries  picked  up
randomly  among   $\set{1,2,3}$,  such  that  the   correct  structure
$\mathcal{A}^\star$ is always composed of 3 groups.  Then, the initial
group  sizes   $n_k$  are   drawn  from  a   multinomial  distribution
$\mathcal{M}(n,  (p_1,  \dots,  p_K))$  with   $p_k  =  1/K$  for  all
$k=1,\dots,K$,  such  that  the   $n_k$  are  approximately  balanced.
Finally,  we let  $\varepsilon_i\sim\mathcal{N}(0,1)$ to  generate the
vector of data $\by = (y_1,\dots,y_n)$.

We compare  the capability of  three weighting schemes to  recover the
true grouping $\mathcal{A}^\star$, namely the fused-ANOVA weights, the
Cas-ANOVA   weights,   and   the  so-called   \emph{default   weights}
corresponding to $w_{k\ell} = n_k  n_\ell$, which are not adaptive but
produce  a path  of solutions  that contains  no split.   Such weights
correspond to the Clusterpath weights  adapted to the ANOVA setup.  We
use our own code for each method.  Typically, the computational burden
required by Cas-ANOVA is huge, compared to the other two procedures as
the  path  of  solutions   may  contain  splits.   Qualitatively,  the
difference  would be  as  in Figure  \ref{fig:timings}, middle  panel.
Thus, we  typically force the  algorithm not  to split when  using the
Cas-ANOVA weights.

We generate data  as specified below, and for each  procedure we check
whether  there exists  at least  one $\lambda$  for which  the correct
structure is identified along the  path of solutions.  The probability
of true support recovery is evaluated by replicating this experiment a
large number  of times ($8096$ times\footnote{this  number arises from
  the  manifold  computer  cores  available.}).   To  investigate  the
asymptotic behavior of  each method, we vary $n$ from  50 to 1,000 and
consider two  scenarios for the  initial number of groups  $K$. First,
$K$ is fixed  at $10$ such that  the number of elements  in each group
grows with $n$. In the second scenario, $K$ grows with $n$ through the
relationship  $K=2.5  \cdot \log(n)$.   The  results  are reported  on
Figure~\ref{fig:consistency},  with  the  first  (resp.   the  second)
scenario on  the left  (resp.  the right)  panel. The  results confirm
Theorem   \ref{thm:oracle_prop}.    The   two   adaptive   procedures,
Cas-ANOVA,  and  to  a   greater  extent,  fused-ANOVA,  dominate  the
non-adaptive weights.  As expected from Section \ref{sec:consistency},
fused-ANOVA  always  dominate  Cas-ANOVA,   as  experienced  in  other
scenarios (\textit{e.g.},  $K=C \cdot \sqrt{n}$) not  reported here to
save space.

\begin{figure}[htbp!]
  \centering
  \begin{tabular}{@{}l@{\hspace{.75em}}ccr}
    & \small $a) \ K = \text{cst.}$ & \small $b) \ K = C\log(n)$ & \\
    \rotatebox{90}{\hspace{4.5em}\small
      $\hat{\prob}(\hat{\mathcal{A}}_n= \mathcal{A}^\star)$} &
    \includegraphics[width=.33\textwidth]{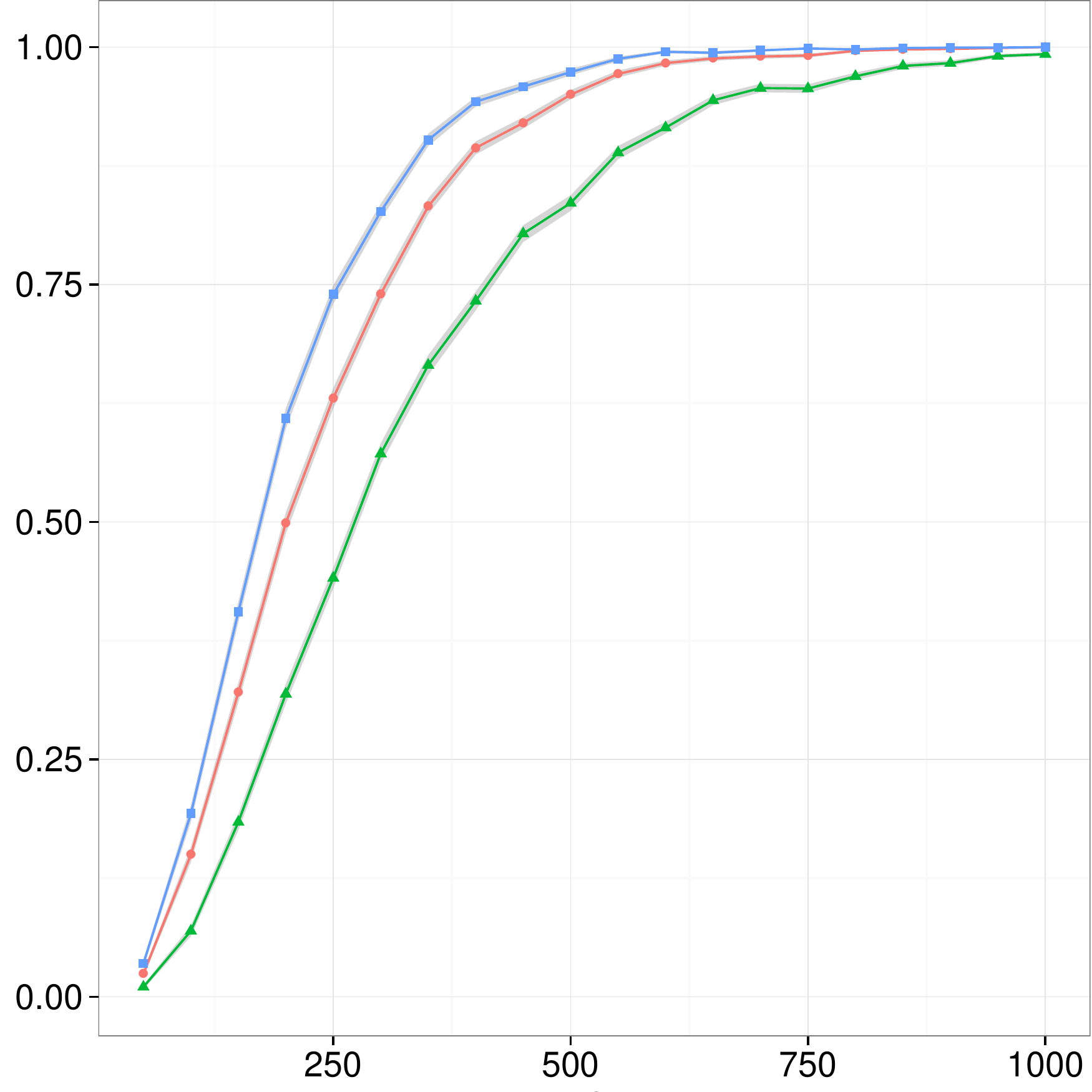} &
    \includegraphics[width=.33\textwidth]{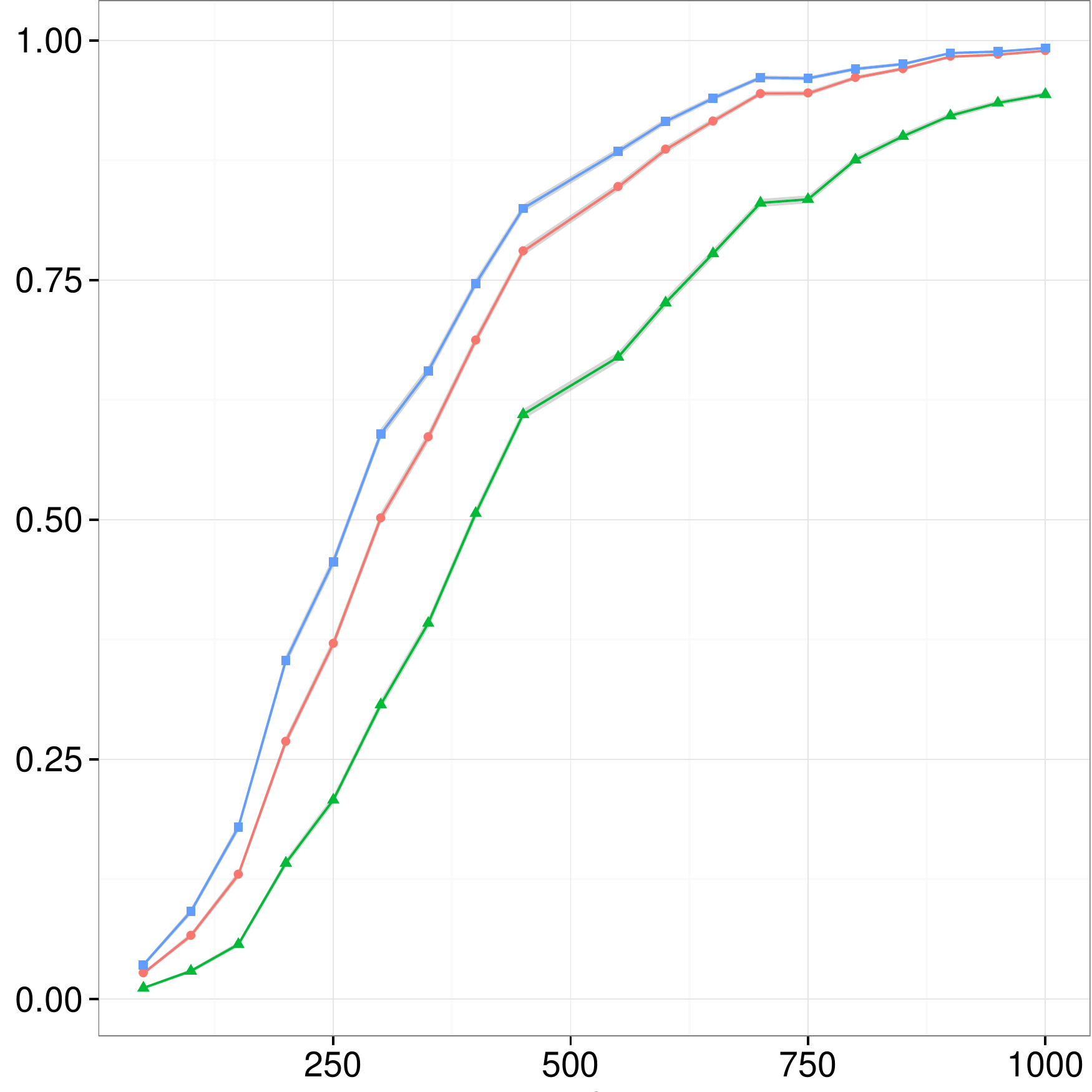} &
    \includegraphics[width=.17\textwidth,clip=true, trim=0 -110pt 0 0]{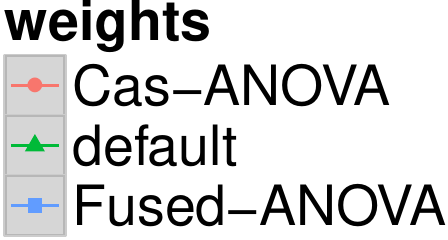} \\
    & \multicolumn{2}{c}{sample size $n$}
  \end{tabular}

  \caption{Univariate case: estimated probability of consistency as a function of the
    sample size  $n$, for  various weights and  in two  scenarios: the
    number of  initial groups $K$ is  either $a)$ fixed  to a constant
    (10)  or $b)$  increases in  $C  \log(n)$ with  $C=2.5$. The  true
    number of groups in $\mathcal{A}^\star$ is 3.}
  \label{fig:consistency}
\end{figure}

\paragraph*{Numerical   illustration    in   the    bivariate   case.}
Theorem~\ref{thm:oracle_prop}  characterizes  the  asymptotic  of  the
fused-ANOVA  estimators  when considering  one  dimension  at a  time.
Concerning the multidimensional setting,  there are two situations. On
the first hand, there exists a dimension such that all the true groups
are   different,   \emph{i.e}   $\beta_{kj}^\star   \neq   \beta_{\ell
  j}^\star$. In  this case,  our theorem  guarantees that,  using this
particular dimension,  the recovered  classification will  converge to
the true one. On the second  hand, there exists no dimension such that
the  true  groups  are  all  different.  In  that  case,  we  have  no
theoretical  guarantee  to support  the  fused-ANOVA  weights.  It  is
nonetheless possible to aggregate  the classification obtained in each
dimension to  a consensus classification.  For a given  $\lambda$, two
individuals $k$ and $\ell$ are in the same multidimensional cluster if
they have been fused on every dimension.

In order  to evaluate empirically  the performance of  the aggregation
step, we consider a two  dimensional classification problem with three
classes and two  scenarii.  Each \emph{prior} group is  drawn from one
of  three classes.   In the  first  scenario, the  three classes  have
different means on the first dimension and the same mean on the second
dimension. The mean vectors are $(1,1.5); (2,1.5); (3,1.5)$, as in top
left panel of Figure~\ref{fig:consistency2D}.  In the second scenario,
both dimensions are informative: the first dimension separates classes
$\set{1,2}$  from  $\set{3}$  while  the  second  dimension  separates
classes  $\set{1,3}$ from  $\set{2}$.   The mean  vectors are  $(1,1);
(1,2);     (2,1)$,      as     in      top     right      panel     of
Figure~\ref{fig:consistency2D}).  We  increase the difficulty  in each
scenario by adding a Gaussian noise with increasing standard deviation
$\sigma$. Results in Figure~\ref{fig:consistency2D} corresponds to the
estimated probability of true  classification recovery along the path,
averaged over 2,000 runs.
\begin{figure}[htbp!]
  \centering
  \begin{tabular}{@{}l@{\hspace{.75em}}ccr}
    & & & \\
    \rotatebox{90}{\hspace{2.5em}\small$\bar{\by}_{k,1}$} & \includegraphics[width=.33\textwidth]{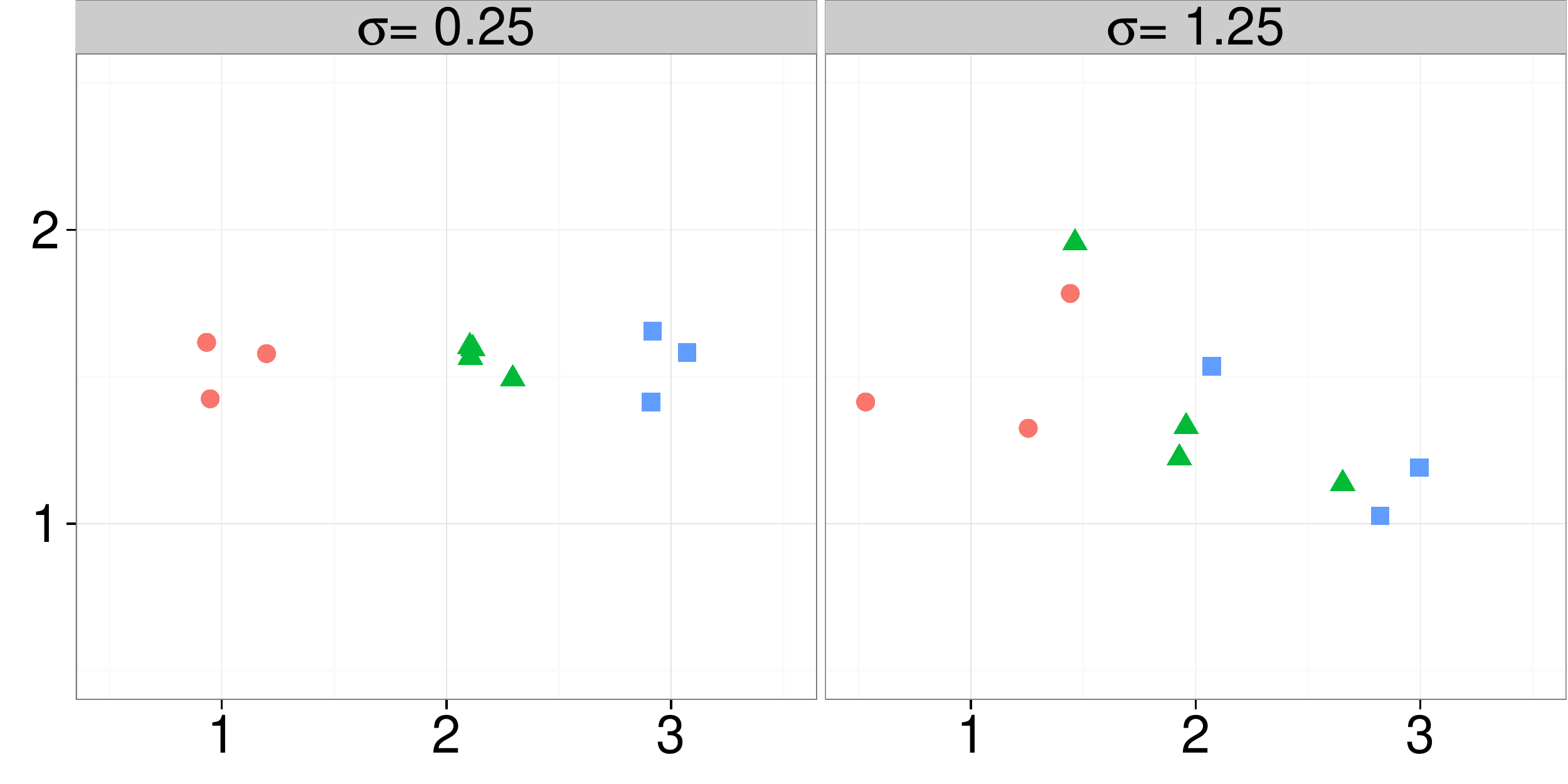} &
    \includegraphics[width=.33\textwidth]{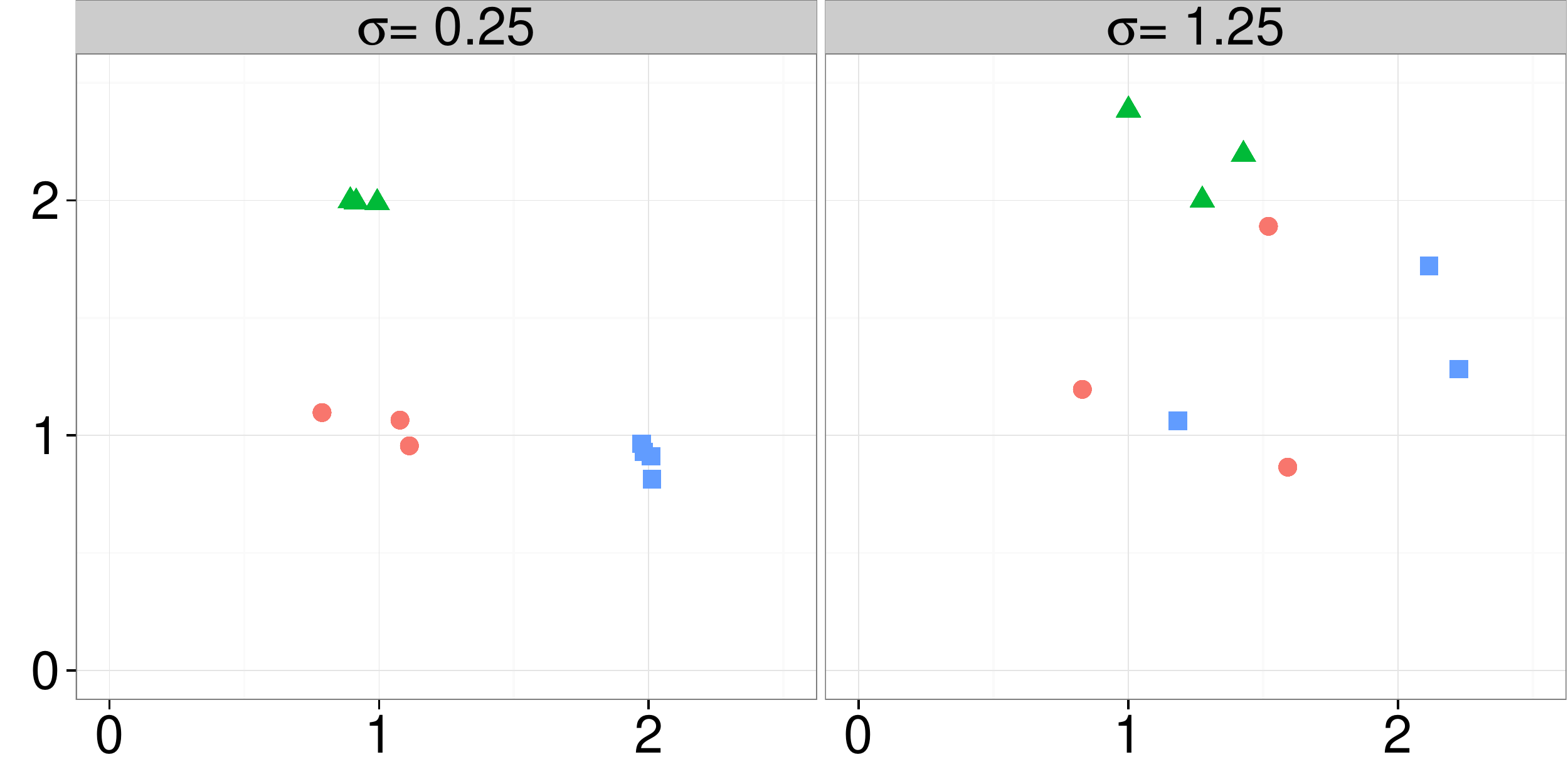} & \\
    & \small$\bar{\by}_{k,2}$ & \small$\bar{\by}_{k,2}$ & \\[2ex]

    \rotatebox{90}{\hspace{4.5em}\small$\hat{\prob}(\hat{\mathcal{A}}_n= \mathcal{A}^\star)$} &
    \includegraphics[width=.33\textwidth]{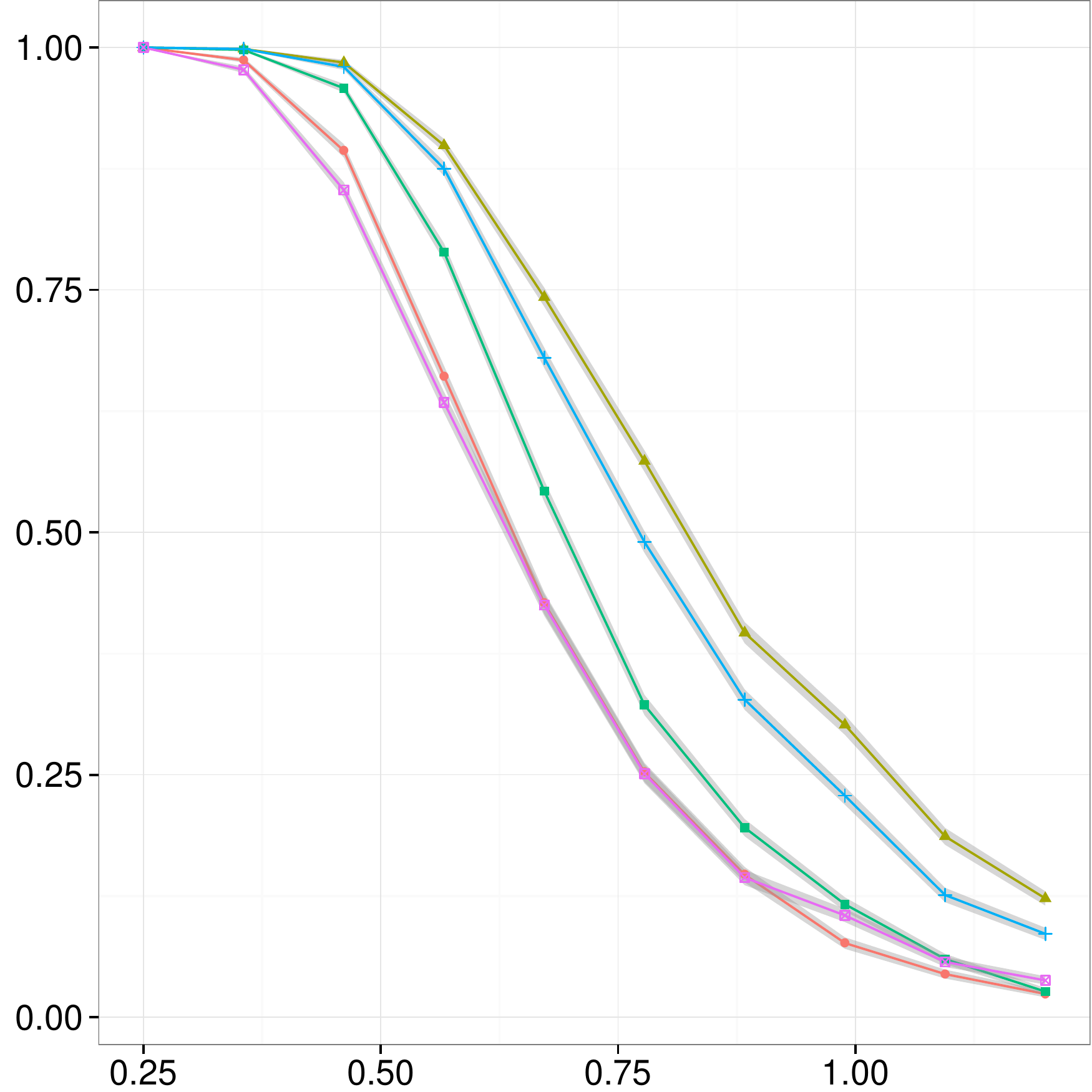}
    &
    \includegraphics[width=.33\textwidth]{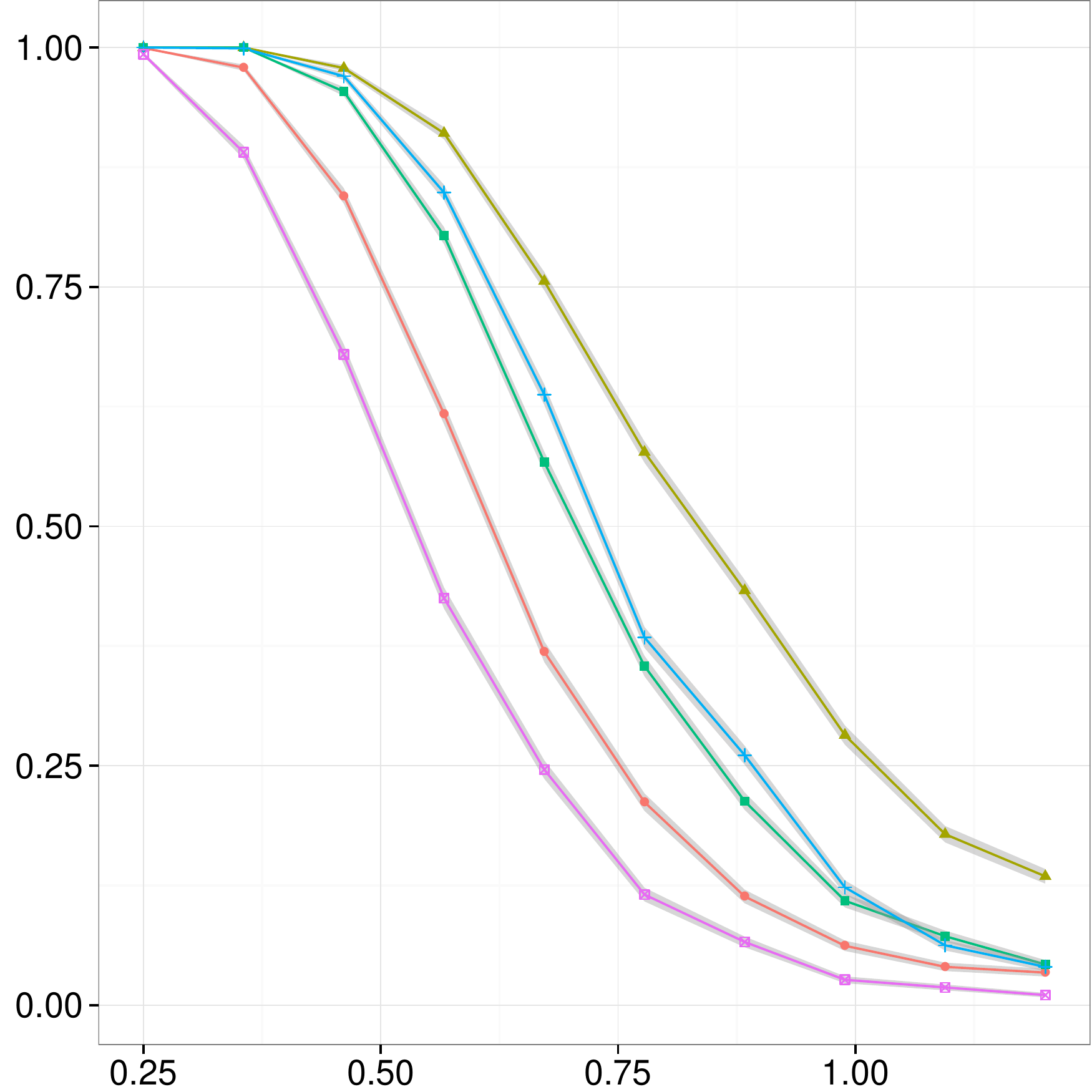} &
    \includegraphics[width=.17\textwidth,clip=true,  trim=0  -40pt  0
    0]{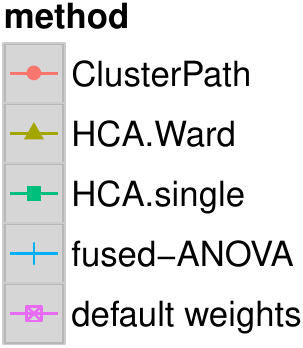}  \\
    & \multicolumn{2}{c}{standard deviation $\sigma$} & \\

   \end{tabular}

  \caption{Bivariate example: estimated  probability of consistency as
    a function of  the noise standard deviation  $\sigma$, for various
    clustering methods. The  initial number  of groups $K$ fixed  to a
    constant (10).   The true number of  groups in $\mathcal{A}^\star$
    is 3.}
  \label{fig:consistency2D}
\end{figure}

In both scenarii, the  fused-ANOVA weights with aggregation outperform
the  multidimensional  $\ell_2$-Clusterpath  as  well  as  the  single
linkage  hierarchical clustering.   The  Ward hierarchical  clustering
shows better performance but at a much higher computational cost.

In this  simple multidimensional numerical study,  we always aggregate
the classification  over the dimensions. However,  this aggregation is
not  necessarily  better  than performing  classification  on  single,
well-chosen feature. We illustrate this point in the following section
on phylogenetic data.


\section{A complete example in phylogeny}
\label{sec:application}

Evolutionary trees  -- sometimes referred  to as ``trees of  life'' --
are  one  of  the  most  emblematic  hierarchical  representations  in
computational biology.   They are  typically used in  phylogenetics to
compare biological  species based on their  similarities regarding one
or several  features.  These  features could  be phenotypic  traits or
genetic  characteristics.    In  these  tree  structures,   each  node
corresponds to a  taxonomic unit, the root node being  the most recent
common ancestor  to all  leaves on the  tree.  All  other intermediate
nodes  between  root  and  leaves represent  the  taxonomic  knowledge
between   the   species   of   interest.   The   study   depicted   in
\cite{phylo_data}   enters  this   framework   by  more   specifically
considering features  associated with  bacterial genomes  to determine
the  phylogenetic  relationships  between  the  taxa.   The  data  set
consists  of  various  genetic   features  associated  with  $n=1,690$
complete  bacterial  genomes  classified in  $K=903$  known  bacterial
species.

We  apply our  method on  this data  set to  assess its  capability of
capturing  the true  underlying  taxonomic structure.   To  do so,  we
consider the  following genetic  features to construct  the hierarchy:
the number of known genes, the number of known proteins and the genome
size (measured  by the  number of  bases in  millions).  We  apply the
univariate model \eqref{eq:criterion_univariate_l1} on each feature to
reconstruct a tree structure.  The indexing function $\kappa$ is built
from  the lowest  level of  classification available  that splits  the
genomes into  $K=903$ bacterial species.   We use the  default weights
and the fused-ANOVA weights \eqref{eq:fa_weights} with $\alpha$ chosen
specifically for each feature (see below).  We also apply hierarchical
clustering  using  Ward's  criterion   and  starting  from  the  known
classification  in  bacterial  species.   Hierarchical  clustering  is
applied  individually on  each feature,  as well  as across  the three
features  using  the  Euclidean   distance  to  build  the  similarity
matrix. To assess the relevance of the inferred trees, we compare them
with various  levels of the  known taxonomic classification  above the
species level, namely  genus (470 groups), family  (216 groups), order
(100 groups), class (46 groups) and  phylum (27 groups).  To this end,
we  compute  the  best  adjusted  rand-index  between  the  respective
reference classifications and the  classifications obtained by cutting
an inferred tree  at all the possible levels of  the hierarchy.  As an
example, we  report in Figure~\ref{fig:phylogenetic_tree} a  subset of
the tree inferred  by the fusion penalty with  fused-ANOVA weights and
the  cutting level  that leads  to the  best performance  in terms  of
adequacy with the true phylum taxonomy.
\begin{figure}[htbp!]
  \centering
  \begin{center}
    \begin{minipage}[c]{.79\linewidth}
      \includegraphics[angle=-90]{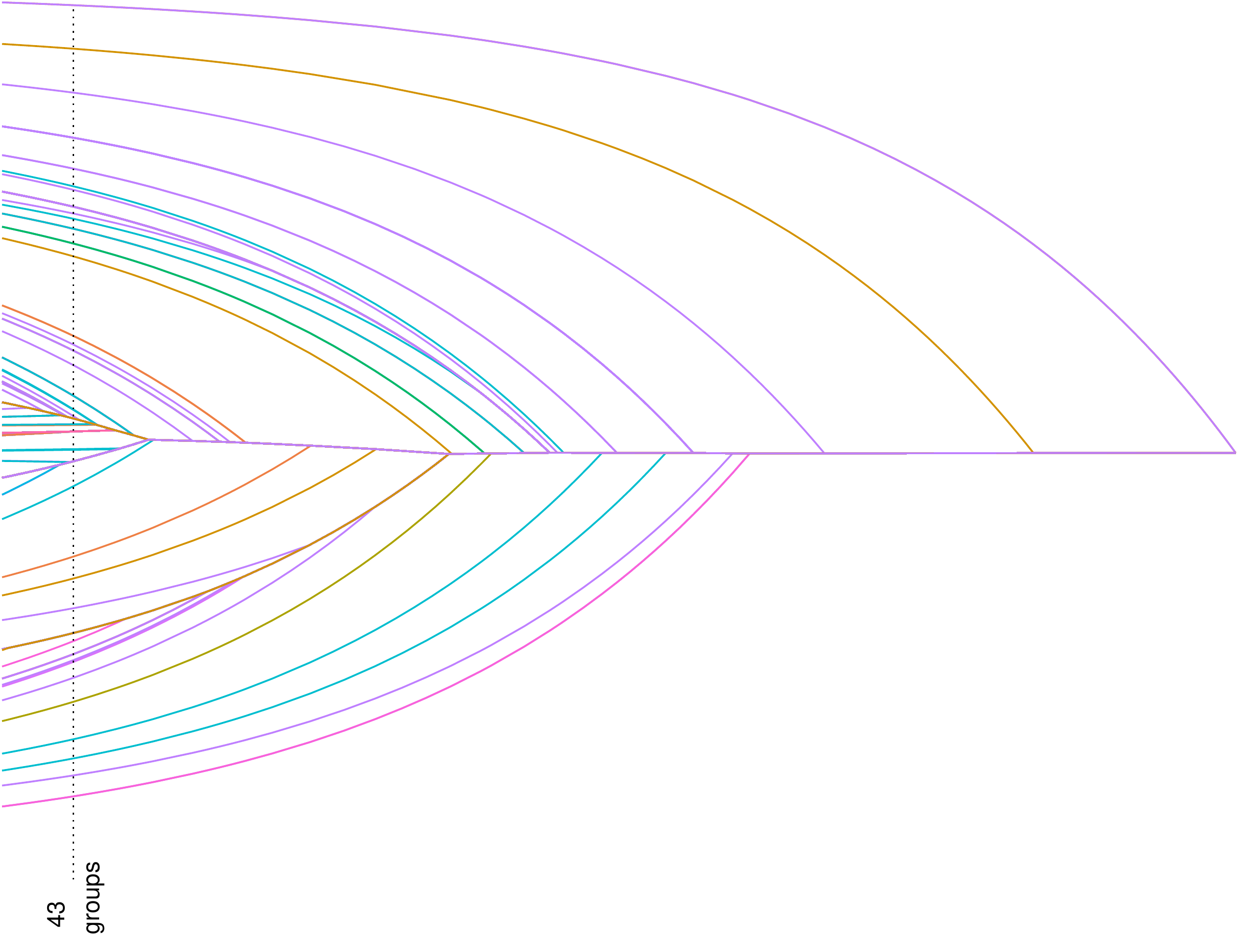} 
    \end{minipage}
    \hspace{-2cm}
    \begin{minipage}[c]{.2\linewidth}
      \includegraphics{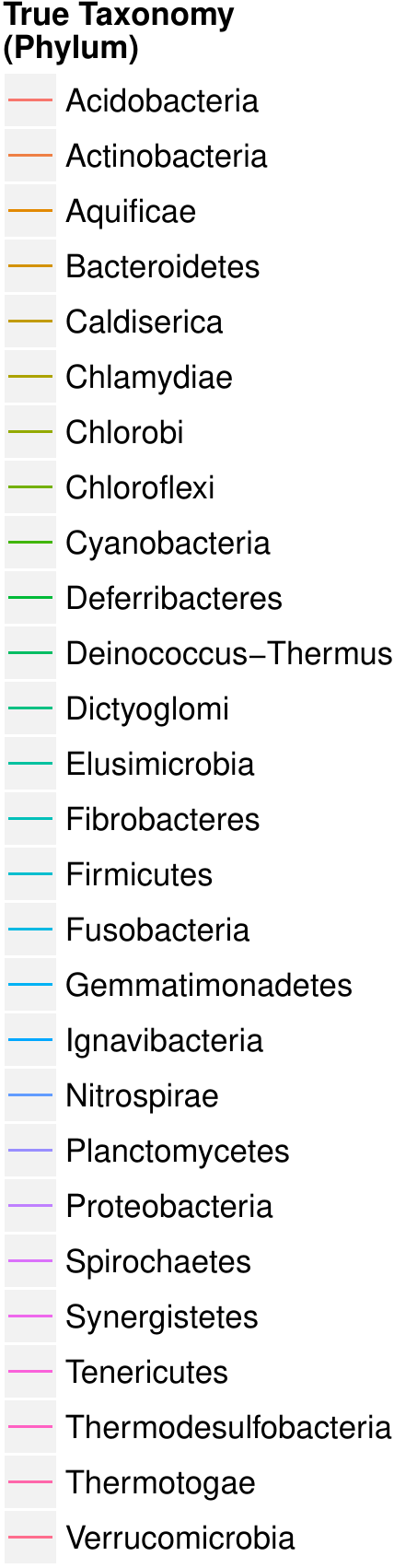}    
    \end{minipage}
  \end{center}
  
  \caption{Tree  reconstructed with  fused-ANOVA from  the ``Size.Mb''
    trait (Adjusted Rand-Index=0.71). Projected colors correspond to
    the true taxonomy (phylum).}
  \label{fig:phylogenetic_tree}
\end{figure}

More quantitative  results are reported  in Figure~\ref{fig:taxonomy},
with the  adjusted rand indexes  for the taxonomic  classifications in
terms of phylum,  order and family, using either the  number of genes,
the number of proteins or the  genome size as the feature variable for
classification.   We  also  represent  the  consensus/multidimensional
classifications either  obtained by  aggregating the  three univariate
fused-ANOVA trees  or by considering  the three features  together for
Ward hierarchical clustering.  Note  that for the fused-ANOVA weights,
we  apply our  method on  a grid  of $\alpha$  and report  the results
obtained for the best $\alpha$ in terms of adjusted rand-index.
\begin{figure}[htbp!]
  \centering
  \begin{tabular}{@{}lc@{}}
    \rotatebox{90}{\hspace{10em}\small Adjusted Rand-Index} &
    \includegraphics[width=.9\textwidth]{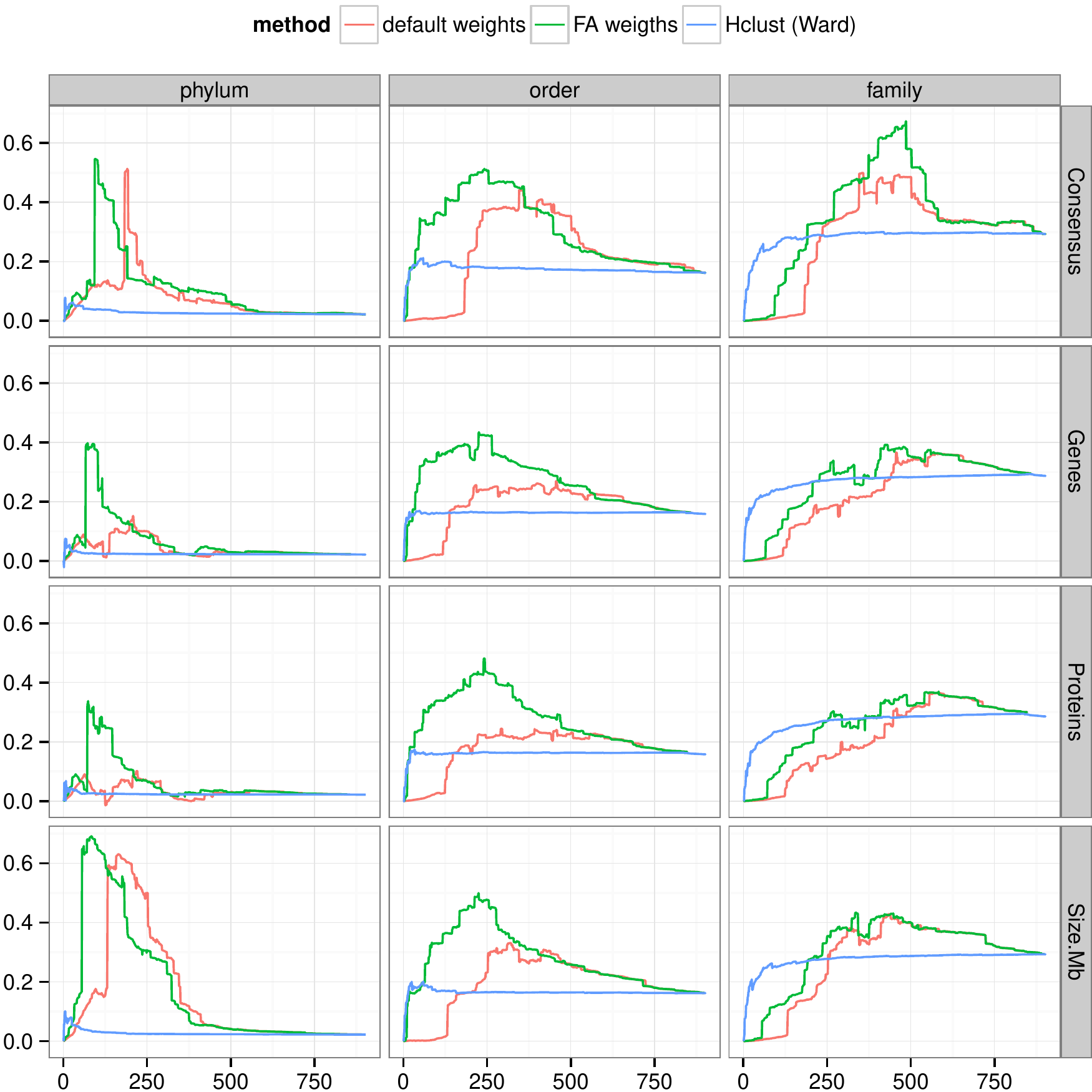} \\
    & \small number of groups \\
  \end{tabular}

  \caption{Adequacy  of  various   tree-based  clustering  methods  to
    different levels of phylogenetic classification.}
  \label{fig:taxonomy}
\end{figure}

First, we notice that the fused-ANOVA weights \emph{always} outperform
the default weights.  This is expected  since the former weights are a
special case of  the latter when $\alpha\to 0$.  Second,  we note that
the consensus  classification -- or  the one obtained  by multivariate
hierarchical clustering  -- is  not always the  best choice.   This is
particularly obvious for the phylum classification, where the ``size''
feature leads  to very good  results in terms of  adjusted rand-index.
These   results    considerably   deteriorate   for    the   consensus
classification, due to  the relatively poor results  obtained from the
``genes'' and ``proteins'' features. Finally, the most striking result
in  Figure~\ref{fig:taxonomy} is  that the  fusion penalty  approaches
clearly outperform the Ward hierarchical clustering.
At first  glance, one might  argue that  the weighting scheme  used in
fused-ANOVA is  responsible for  such good performance.   However, the
fusion  penalty with  default  weights remains  competitive  in a  few
cases.  This  supports the  fact that the  regularizing virtue  of the
fusion penalty is of great help when the problem size is high.


\section*{Acknowledgments}

We  would like to  thank Mahendra  Mariadassou for  useful discussions
about the bacterial genomes data  set and for sharing his knowledge of
phylogeny.

\bibliography{biblio_fa}

\begin{thebibliography}{17}
\providecommand{\natexlab}[1]{#1}
\providecommand{\url}[1]{\texttt{#1}}
\expandafter\ifx\csname urlstyle\endcsname\relax
  \providecommand{\doi}[1]{doi: #1}\else
  \providecommand{\doi}{doi: \begingroup \urlstyle{rm}\Url}\fi

\bibitem[Bondell and Reich(2008{\natexlab{a}})]{Oscar2008}
H.~D. Bondell and B.~J. Reich.
\newblock Simultaneous regression shrinkage, variable selection, and supervised
  clustering of predictors with {OSCAR}.
\newblock \emph{Biometrics}, 64\penalty0 (1):\penalty0 115--123,
  2008{\natexlab{a}}.

\bibitem[Bondell and Reich(2008{\natexlab{b}})]{bondell2008simultaneous}
H.D. Bondell and B.J. Reich.
\newblock Simultaneous factor selection and collapsing levels in {ANOVA}.
\newblock \emph{Biometrics}, 65\penalty0 (1):\penalty0 169--177,
  2008{\natexlab{b}}.

\bibitem[Boyd and Vandenberghe(2004)]{cvx_optim}
S.P. Boyd and L.~Vandenberghe.
\newblock \emph{Convex optimization}.
\newblock Cambridge Universuity Press, 2004.

\bibitem[Efron et~al.(2004)Efron, Hastie, Johnstone, and
  Tibshirani]{efron2004least}
B.~Efron, T.~Hastie, I.~Johnstone, and R.~Tibshirani.
\newblock Least angle regression.
\newblock \emph{Ann. Statist.}, 32\penalty0 (2):\penalty0 407--499, 2004.

\bibitem[Fan and Li(2001)]{2001_JASA_fan}
J.~Fan and R.~Li.
\newblock {Variable Selection via Nonconcave Penalized Likelihood and Its
  Oracle Properties}.
\newblock \emph{JASA}, 96\penalty0 (456):\penalty0 1348--1360, 2001.

\bibitem[Friedman et~al.(2007)Friedman, Hastie, H{\"o}fling, Tibshirani,
  et~al.]{friedman2007pathwise}
Jerome Friedman, Trevor Hastie, Holger H{\"o}fling, Robert Tibshirani, et~al.
\newblock Pathwise coordinate optimization.
\newblock \emph{The Annals of Applied Statistics}, 1\penalty0 (2):\penalty0
  302--332, 2007.

\bibitem[Fu and Knight(2001)]{FuKnight2001}
W.~Fu and K.~Knight.
\newblock Asymptotics for lasso-type estimators.
\newblock \emph{Ann. Statist.}, 28\penalty0 (5):\penalty0 1245--1501, 2001.

\bibitem[Gertheiss and Tutz(2010)]{gertheiss2010sparse}
Jan Gertheiss and Gerhard Tutz.
\newblock Sparse modeling of categorial explanatory variables.
\newblock \emph{The Annals of Applied Statistics}, pages 2150--2180, 2010.

\bibitem[Geyer(1994)]{Geyer1994}
C.J. Geyer.
\newblock Simultaneous factor selection and collapsing levels in {ANOVA}.
\newblock \emph{Ann. Statist.}, 22\penalty0 (4):\penalty0 1635--2134, 1994.

\bibitem[Hocking et~al.(2011)Hocking, Vert, Bach, and
  Joulin]{HOCKING-clusterpath}
T.~Hocking, J.-P. Vert, F.~Bach, and A.~Joulin.
\newblock Clusterpath: an algorithm for clustering using convex fusion
  penalties.
\newblock In \emph{Proceedings of the 28th ICML}, pages 745--752, 2011.

\bibitem[Hoefling(2010)]{hoefling2010path}
H.~Hoefling.
\newblock A path algorithm for the fused lasso signal approximator.
\newblock \emph{J. Comput. Graph. Statist.}, 19\penalty0 (4):\penalty0
  984--1006, 2010.

\bibitem[Mairal and Yu(2012)]{ICML2012Mairal_202}
J.~Mairal and B.~Yu.
\newblock Complexity analysis of the lasso regularization path.
\newblock In \emph{Proceedings of the 29th ICML}, 2012.

\bibitem[Rosset and Zhu(2007)]{2007_AS_rosset}
S.~Rosset and J.~Zhu.
\newblock Piecewise linear regularized solution paths.
\newblock \emph{Ann. Statist.}, 35\penalty0 (3):\penalty0 1012--1030, 2007.

\bibitem[Tibshirani and Taylor(2011)]{2011_AS_Tibshirani}
R.~J. Tibshirani and J.~Taylor.
\newblock The solution path of the generalized lasso.
\newblock \emph{Ann. Statist.}, 39\penalty0 (3):\penalty0 1335--1371, 2011.

\bibitem[Vetrovsky and Baldrian(2013)]{phylo_data}
Tomas Vetrovsky and Petr Baldrian.
\newblock The variability of the {16S rRNA} gene in bacterial genomes and its
  consequences for bacterial community analyses.
\newblock \emph{PLoS ONE}, 8\penalty0 (2):\penalty0 e57923, 2013.
\newblock \doi{10.1371/journal.pone.0057923}.

\bibitem[Viallon et~al.(2014)Viallon, Lambert-Lacroix, Hoefling, and
  Picard]{viallon:hal-00813281}
Vivian Viallon, Sophie Lambert-Lacroix, H{\"o}lger Hoefling, and Franck Picard.
\newblock On the robustness of the generalized fused lasso to prior
  specifications.
\newblock \emph{Statistics and Computing}, pages 1--17, 2014.

\bibitem[Zou(2006)]{zou2006adaptive}
H.~Zou.
\newblock The adaptive lasso and its oracle properties.
\newblock \emph{JASA}, 101\penalty0 (476):\penalty0 1418--1429, 2006.

\end{thebibliography}
\bibliographystyle{plainnat}

\appendix


\section{Proofs}
\subsection{Theorem \ref{thm:norm-nosplit} (absence of splits with norms)}
\label{sec:norm-nosplit}

For the sake  of brevity the proof is detailed  only in the clustering
framework, \textit{i.e.},  when $\kappa(i)  = i$  and $n_k=1$  for all
$k=1,\dots,K$.   The  generalization  to  groups with  more  than  one
individual is straightforward and follows the exact same line.

Consider the objective function  in \eqref{eq:criterion_general} and a
time $\lambda_0$  at which  we have  a valid set  of clusters.   It is
obvious that clusters containing only one individual cannot split.  We
will thus consider  clusters grouping together more  than one element.
We  denote by  $C=\set{k:\beta_k (\lambda_0)  = \beta_{C}(\lambda_0)}$
such  a cluster,  with $\beta_{C}$  the current  estimated mean.   For
unitary weights, Lemma \ref{lem:KKT_sum_beta} implies that
\[
  \bzr_p = - \bar{\by}_C + \bbeta_C + \lambda_0 \sum_{i \notin C} \frac{\partial \Omega(\bbeta_C - \bbeta_{i})}{\partial \bbeta_C}(\lambda_0).
\]

Subtracting the above equation from the subgradient equation
\eqref{eq:KKT_general} for $i\in C$, one has
\begin{equation}
  \label{eq:norm_nosplit_proof}
\bar{\by}_C - \by_i + \lambda_0 \sum_{j \in C} \btau_{ij}(\lambda_0) =
\bzr_p.
\end{equation}

We now  consider any time $\lambda  \geq \lambda_0$ such  that no fusion
has   occurred.   Let   us  show   that  for   $\btau_{ij}(\lambda)  =
\frac{\lambda_0}{\lambda}  \btau_{ij}(\lambda_0)$, it  is  possible to
solve the KKT conditions, and thus show that no split occurs.

First,  the proposed $\btau_{ij}(\lambda)$  are valid  subgradients as
$\Omega(\btau_{ij}(\lambda))        \leq        1       $        since
$\Omega(\btau_{ij}(\lambda_0))   \leq   1$  and   $\lambda>\lambda_0$.
Second,  for  this particular  choice  of  subgradient  and thanks  to
\eqref{eq:norm_nosplit_proof}, the KKT conditions for all $C$ and all
$i\in C$ simplify as follows:
\begin{multline*}
  \bbeta_C - \by_i +  \lambda \sum_{j \in C} \frac{\lambda_0}{\lambda}
  \btau_{ij}(\lambda_0) + \lambda \sum_{C' \neq C} |C'| \frac{\partial
    \Omega(\bbeta_C - \bbeta_{C'})}{ \partial \bbeta_C}(\lambda) \\
  = \bbeta_C -\bar{\by}_C + \lambda \sum_{C' \neq C} |C'| \frac{\partial
    \Omega(\bbeta_C - \bbeta_{C'})}{ \partial \bbeta_C}(\lambda).
\end{multline*}

It now remains to check that we  can find a $\bbeta$ which zeroes this
subgradient equation.  Note that for  all $C'\neq C$, the differential
$\partial   \Omega(\bbeta_C  -   \bbeta_{C'})   /  \partial   \bbeta_C
(\lambda)$ is well defined.  Then, by multiplying the above expression
by $|C|$, we obtain the gradient of the following objective function
\[ \frac{1}{2}  \sum_{i=1}^n \left\|  \by_i - \bbeta_i  \right\|_2^2 +
\lambda  \sum_{C,C':C \neq  C'}  |C| \cdot  |C'|  \ \Omega(\bbeta_C  -
\bbeta_{C'}).  \]

This is a strictly convex  problem admitting one unique solution which
is solved by zeroing its gradient. Thus we necessarily have
\begin{equation*}
  \bbeta_C -\bar{\by}_C + \lambda \sum_{C' \neq C} |C'| \frac{\partial
    \Omega(\bbeta_C - \bbeta_{C'})}{ \partial \bbeta_C}(\lambda) = \bzr_p,
\end{equation*}
which ends the proof.


\subsection{Theorem  \ref{thm:weight-nosplit}: absence of  splits with
  distance-decreasing weights in 1-d}
\label{sec:nosplit}

For the sake of  brevity the proof is detailed only  in the case where
$\kappa(i) = i$ and $n_k=1$ for all $k=1,\dots,K$.  The generalization
to groups  with more  than one  individual is  straightforward, seeing
that we can replace a group $\kappa(i)$ by $n_{\kappa(i)}$ individuals
with  value  $\sum_{j} y_{j}  /  n_{\kappa(i)}$.   Also, when  $\Omega
\equiv \ell_1$, the proof remains  valid but should be done separately
on each dimension.

Throughout the proof, we may thus consider the estimator defined by
\begin{equation}
  \label{eq:crit_proof}
  \hatbbeta(\lambda) = \argmin_{\bbeta \in \Rset^n} \frac{1}{2} \sum_{i=1}^n \left( y_i -
    \beta_i \right)^2
  + \lambda \ \sum_{i,j:i\neq j} w_{ij} |\beta_i - \beta_j|.
\end{equation}

The   proof   proceeds  in   two   steps   detailed  hereafter:
\begin{enumerate}
\item in subsection \ref{subsec:order}, we show that absence of splits
  is equivalent to preservation of the order along the path;
\item   in  subsection  \ref{subsec:dual},   we  show   that  distance
  decreasing  weights  preserve  the  order,  by  considering  a  dual
  formulation of Problem \eqref{eq:crit_proof}.
\end{enumerate}

For simplicity,  we consider that  the data vector $\by$  is initially
ordered such that
\[y_1 \geq \dots, y_i \geq y_{i+1} \geq \dots \geq y_n.\]

\subsubsection{Preserving the order}
\label{subsec:order}

We say  that the  loss is order-preserving,  if $y_i\leq  y_j$ implies
that  $\hat{\beta}_i(\lambda)  \leq \hat{\beta}_j(\lambda)$,  for  all
$\lambda\geq 0$.

\begin{lemma}
  \label{lem:nosplit_order}  The  absence  of  splits is  equivalent  to
  preservation   of   the   order   along   the   path   for   Problem
  \eqref{eq:crit_proof}.
\end{lemma}

\begin{proof}
  First of all, in the absence of splits in the path, it is clear that
  the order is preserved.

  Conversely, assume that there is an event at $\lambda_0$ that splits
  a  group  $C$  into  $C_{\text{down}}$  and  $C_{\text{up}}$,  where
  $\hatbeta_{\text{down}}(\lambda)   <  \hatbeta_{\text{up}}(\lambda)$
  for   all  $\lambda   \geq   \lambda_0$.   By   means  of   Equation
  \eqref{eq:KKT_sum_beta},         we         necessarily         have
  $\bar{y}_{{\text{down}}} \geq  \bar{y}_{{\text{up}}}$ as illustrated
  on  Figure   \ref{fig:slope_split}.   However,   if  the   order  is
  preserved, for all  $(i,j)\in C_{\text{down}} \times C_{\text{up}}$,
  we   must   have   $y_i<   y_j$   and   $\bar{y}_{{\text{down}}}   <
  \bar{y}_{{\text{up}}}$, which leads to a contradiction.
\end{proof}

\begin{figure}[htbp!]
  \centering
  \includegraphics[width=4cm]{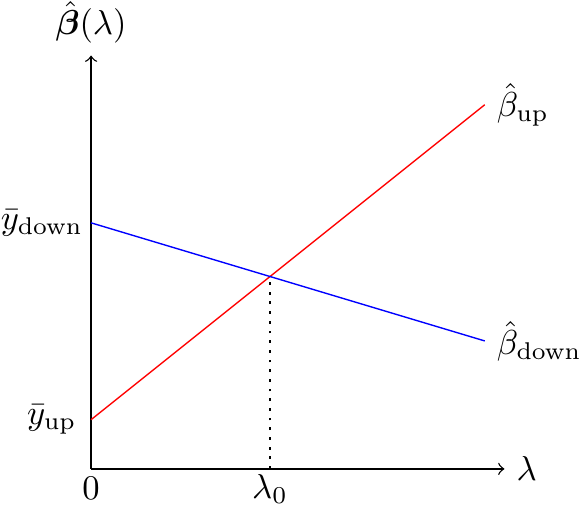}
  \caption{Equivalence  between preserving  the order  and absence  of
    splits relies on a simple geometrical argument.}
  \label{fig:slope_split}
\end{figure}

\subsubsection{The dual problem}
\label{subsec:dual}

We  follow arguments  developed by  \cite{2011_AS_Tibshirani}  for the
generalized Lasso. Indeed, Problem \eqref{eq:crit_proof} can be recast
as
\begin{equation}
  \label{eq:crit_genlasso}
  \hatbbeta(\lambda)   =  \argmin_{\bbeta  \in   \Rset^n}  \frac{1}{2}
  \left\| \by - \bX \bbeta \right\|_2^2 + \lambda \|W D \bbeta \|_1,
\end{equation}
a  generalized Lasso  problem with  $\bX =  \bI_{nn}$, $W$  a diagonal
matrix  the  diagonal   of  which  is  the   $n(n-1)/2$  vector  given
by                          \[\text{diag}(W)                         =
(w_{11},\dots,w_{1n},w_{23},\dots,w_{2n},w_{34},\dots,w_{(n-1)n})\]
and $D$  is a $n(n-1)/2  \times n$  matrix that performs  the pairwise
differences such that

\begin{equation}
  \label{eq:theDmatrix}
  D = \kbordermatrix{
  \text{couple } (i,j) & & & & & \\
  (1,1) & 1 & -1 & & & \\
  (1,2) & 1 & & -1 & & \\
  \vdots & \vdots & & & \ddots & \\
  (1,n) & 1 & & & & -1 \\
  (2,3) & & 1 & -1 & & \\
  \vdots & & \vdots & & \ddots & \\
  (2,n) & & 1 & & & -1 \\
  \vdots & & & &  &  \\
  (n-1,n) & & & & 1 & -1 \\
}.
\end{equation}
In what follows, it will be convenient to index rows of the matrix $D$
in  terms   of  the   couple  $(i,j)$,  as   is  done   in  Expression
\eqref{eq:theDmatrix}.

We  then   rely  on  the   Lagrangian  dual  of  the   primal  problem
\eqref{eq:crit_genlasso}  studied in  \cite{2011_AS_Tibshirani}, which
is
\begin{equation}
  \label{eq:crit_dual}
  \hat{\bu}(\lambda)     =    \argmin_{\bu\in\Rset^{(n(n-1)/2)}}    \frac{1}{2}
  \left\|\by - D^TW 
    \bu\right\|_2^2 \quad \text{subject to } \|\bu\|_\infty \leq \lambda,
\end{equation}
and where the correspondence between the primal and dual variables is
\begin{equation*}
  \hatbbeta = \by - D^TW \hat{\bu}.
\end{equation*}
The dual solution must satisfies
\begin{equation*}
  \hat{u}_{ij} \in \begin{cases}
    \set{+\lambda} & \text{ if } (W D\hatbbeta)_{ij} > 0, \\
    \set{-\lambda} & \text{ if } (W D\hatbbeta)_{ij} < 0, \\
    [-\lambda,+\lambda] & \text{ if } (W D\hatbbeta)_{ij} = 0,\\
  \end{cases}
\end{equation*}
where we  use the indexing in  terms of $(i,j)$ for  the vector $\bu$.
We also define $\mathcal{B}$, the set of $(i,j)$ such that $|u_{ij}| =
\lambda$, that is, the ones reaching the boundary in the dual.

The key point is to note that  the order is not preserved \emph{if and
  only if}, at  some point of the path, there  exists some $(i,j)$ and
$\lambda$ such  that $\hat{u}_{ij}(\lambda) = -\lambda$,  meaning that
$(WD\beta)_{ij} < 0$.  The rest of the proof will show that this event
is not  possible for  distance decreasing weights  and the  matrix $D$
given  by   \eqref{eq:theDmatrix}.   To   this  end,  we   proceed  by
contradiction, by supposing that the  order is not preserved along the
path. We  thus consider the  first split  event that will  disrupt the
order,  which occurs  at $\lambda_0$.   At  this point,  the order  is
preserved and there  is an $\varepsilon> 0$ such  that on $]\lambda_0,
\lambda_0+\varepsilon]$,  the order  is not  preserved.  We  note that
$\lambda_0>0$ since the order is necessarily preserved up to the first
fusion  event  that  fuses  data points  with  different  values.   At
$\lambda_0$, we must have a couple $(i^0,j^0)$ such that $\hat{u}_{i^0
  j^0}(\lambda_0) = -\lambda_0$ that  reaches the boundary.  Moreover,
the  left derivative  $\partial^-  \hat{u}_{i^0j^0}(\lambda)$ must  be
less     than    $-1$     because    the     path    is     continuous
\citep[see][]{2011_AS_Tibshirani}  and because  we consider  the first
event disrupting the order.  We  provide geometrical insight into this
point on Figure \ref{fig:dual_derivative}.
\begin{figure}[htbp!]
  \centering
  \includegraphics[width=4cm]{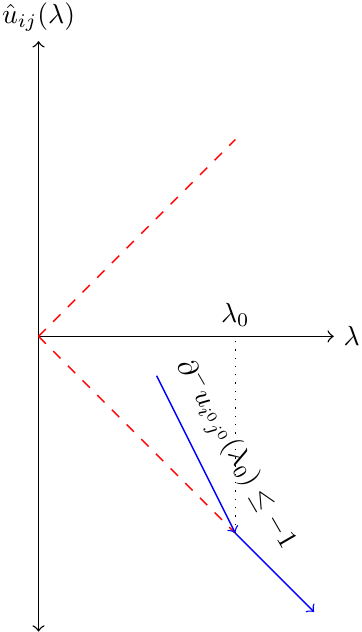}
  \caption{Geometrical insight into a split event in the dual.}
  \label{fig:dual_derivative}
\end{figure}

We    now    show     that    we    necessarily    have    $\partial^-
\hat{u}_{i^0j^0}(\lambda)>-1$,  leading to  a contradiction.   To this
end, we consider the set $\mathcal{C}$ of indices which are fused with
$i^0$  and $j^0$  just  before $\lambda_0$,  that  is, $\mathcal{C}  =
\set{i:\hatbeta_{i} = \hatbeta_{i^0} = \hatbeta_{j^0}}$. We denote by
\[ \Din = \set{ (i,j) \in \mathcal{C}\times \mathcal{C}: i<j }\]
the set of intra $\mathcal{C}$ differences and
\[\Dout  =  \set{(i,j)  \in  \mathcal{B}  :  i<j,  i  \in  \mathcal{C}
  \text{ or  } j  \in \mathcal{C}},\] the  set of  differences between
$\mathcal{C}$ and other groups. Finally we denote by $\Dother$ the set
of all other indices which are  not in $\Dout$ and $\Din$. Given those
sets we reindex the matrix $D$ and $W$ as follows.
\begin{equation*}
  D = \kbordermatrix{
    \text{set of index }   & \Cgr & \bar{\Cgr}  \\
    \Din & D_{\Din \times \Cgr} & D_{\Din\times \bar{\Cgr}} \\
    \Dout & D_{\Dout \times \Cgr} & D_{\Dout\times \bar{\Cgr}} \\
    \Dother & D_{\Dother \times \Cgr} & D_{\Dother \times \bar{\Cgr}} \\
  }
\end{equation*}

\begin{equation*}
  W = \kbordermatrix{
  \text{set of index }   & \Din & \Dout & \Dother  \\
\Din & W_{\Din^2} & 0 & 0 \\
\Dout & 0 & W_{\Dout^2} & 0 \\
\Dother & 0 & 0 & W_{\Dother^2} \\
}
\end{equation*}

By definition, for all $(i,j)\in\Dother$, $i$ and $j$ do not belong to
$\Cgr$  and  thus $D_{\Dother  \times  \Cgr}=  0$.   By simple  matrix
algebra,  the restriction  of $D^T  W$ to  the rows  in $\Cgr$  can be
written
\[(D^T  W)_{\Cgr}  = (D_{\Din  \times  \Cgr})^T  W_{\Din^2} +  (D_{\Dout
  \times \Cgr})^T W_{\Dout^2}.\]

Just        before        $\lambda_0$        and        for        all
$(i,j)\in\mathcal{D}_{\text{out}}$,  we have  $\hat{u}_{ij}(\lambda) =
\lambda$ and  so $\text{sign}(\hat{u}_{ij}) =1$, because  the order is
preserved at  this point.  Then,  following \cite{2011_AS_Tibshirani},
the   KKT    conditions   of   \eqref{eq:crit_dual}    restricted   to
$\mathcal{D}_{\text{in}}$ imply that
\begin{multline}
  \label{eq:expr_uCC}
  \hat{\bu}_{\mathcal{D}_{\text{in}}}(\lambda) = \\
  \left(W_{\Din^2} D_{\Din\times \Cgr} D_{{\Din\times \Cgr}}^T
    W_{\Din^2}\right)^+ W_{\Din^2}D_{{\Din\times \Cgr}}
  \left(\by - \lambda (W_{\Dout^2}D_{{\Dout\times \Cgr}})^T \I_{\Dout} \right),
\end{multline}
where $A^+$ denotes the Moore-Penrose pseudo-inverse of $A$. Note that
such   a    choice   is    important   since   it    guarantees   that
$\hat{\bu}(\lambda)$ is a continuous function of $\lambda$.

Expression of \eqref{eq:expr_uCC}  greatly simplifies by exploiting an
explicit formula for  the pseudo-inverse, which we  derive in the next
paragraph. 

\paragraph*{Analytic form of the  pseudo-inverse.}  In this paragraph,
we  consider   only  the  $D_{\Din   \times  \Cgr}$  matrix   and  the
$W_{\Din^2}$  matrix, which  correspond  to the  set  of intra  $\Cgr$
differences and  their weights.  For  simplicity, we just  denote them
$D$ and $W$ here and call $n'$ the group size. We have
\begin{equation*}
  D^T D = n' \bI_{n'n'} - \I_{n'} \I_{n'}^T, \quad D \I_{n'} = \bzr_{n'},
\end{equation*}
and  from  this we  get  $D  D^T  D =  n'  D$  and  thus, $D^+  =  D^T
/n'$. Finally
\begin{equation*}
  (D D^T)^+  = \frac{1}{{n'}^2} D D^T, \quad  \text{ and } (D  D^T)^+ D =
  \frac{D}{n'} .
\end{equation*}
If    we   now    consider   the    weighted   version    of   Problem
\eqref{eq:crit_genlasso}, one has
\begin{equation*}
  (W D D^T W )^+ W D = W^{-1} (D D^T)^+ W^{-1} W D = \frac{W^{-1}D}{n'}.
\end{equation*}

\paragraph*{Back  to  our   problem,}  Expression  \eqref{eq:expr_uCC}
becomes
\begin{equation}
  \label{eq:expr_uCC_simple}
  \hat{\bu}_{\mathcal{D}_{\text{in}}}(\lambda) =
  \frac{1}{n_{\Cgr}}W_{\Din^2}^{-1} D_{\Din\times\Cgr}
  \left(\by - \lambda \underbrace{(W_{\Dout^2}D_{{\Dout\times \Cgr}})^T \I_{\Dout}}_{V} \right).
\end{equation}

Let us  consider the  size-$n_{\Cgr}$ vector  $V$, which  includes the
differences   between  elements   in  $\Cgr$   and  elements   outside
$\Cgr$. Note that the $i$th column of $D_{{\Dout\times \Cgr}}$ is zero
everywhere, except for the elements  of $\Dout$ containing $i$. In the
last  case,  it  is  equal  to  $1$  if  $y_i\geq  y_j$  and  to  $-1$
otherwise. Hence,
\begin{equation*}
  V_i = \left( (D_{{\Dout\times \Cgr}})^T W_{\Dout^2} \I_{\Dout}\right)_i = \sum_{j\in\bar{\Cgr}} w_{ij} \sign(y_i - y_j).
\end{equation*}

Also recall that the matrix $D_{\Din\times \Cgr}$ encodes the pairwise
positive differences. Then for $y_i > y_{i'}$, the $(i,i')$ element of
$D_{\Din\times \Cgr} V$ equals
\begin{equation*}
  (D_{\Din\times   \Cgr}    V)_{ii'}   =   V_i   -    V_{i'}   =
  \sum_{j\in\bar{\Cgr}} w_{ij} \sign(y_i - y_j) - w_{i'j} \sign(y_{i'} - y_j).
\end{equation*}
There  are  two  possibilities:  either  $y_{j}>y_{i}\geq  y_{i'}$  or
$y_{i}\geq y_{i'}> y_{j}$.   We thus split the summation  in the above
equation into two parts:
\begin{equation*}
  (D_{\Din\times   \Cgr}    V)_{ii'}   =  
  \sum_{\substack{j\in\bar{\Cgr}\\ y_{j}>y_{i}\geq y_{i'}}} (w_{i'j} - w_{ij})
  + \sum_{\substack{j\in\bar{\Cgr}\\y_{i}\geq y_{i'}>y_j}} (w_{ij} - w_{i'j}).
\end{equation*}
And from  this we see  that if the  weights are positive  and distance
decreasing, all  the $(D_{\Din\times \Cgr} V)_{ii'}$  are negative. To
conclude,  the slopes  in Expression  \eqref{eq:expr_uCC_simple}, that
is,
\begin{equation*}
  -  \frac{\lambda }{n_{\Cgr}}W_{\Din^2}^{-1} D_{\Din\times\Cgr}(W_{\Dout^2}D_{{\Dout\times \Cgr}})^T \I_{\Dout}
\end{equation*}
are   positive,   which   is   in   contradiction   with   $\partial^-
\hat{u}_{i^0j^0}(\lambda)\leq -1$.


\subsection{Theorem     \ref{thm:oracle_prop}:     consistency     for
  exponentially adaptive weights}
\label{sec:oracle_prop}

We  essentially follow  the same  line as  for the  adaptive Lasso  in
\cite{zou2006adaptive},  yet  adapted  to  the fusion  penalty  as  in
\cite{viallon:hal-00813281,bondell2008simultaneous}.      The     main
difference comes  from the use  of the exponentially  adaptive weights
$w_{k\ell}^{\text{FA}}$.
\\

We  start  by  asymptotics  in the  vein  of  \cite{FuKnight2001}  for
Lasso-type  procedures:  Lemma  \ref{lem:asym_distr} below  gives  the
limiting     distribution     of     the     fused-ANOVA     estimator
\eqref{eq:fused_anova}  on  the  range  of interest  for  the  penalty
$\lambda_n$ which essentially proves  the asymptotic normality part of
the Theorem.

\begin{lemma}   \label{lem:asym_distr}   Suppose  $\lambda_n   n^{3/2}
  \exp\set{-\alpha\sqrt{n}} \to 0$ and $\lambda_n n^{3/2} \to \infty$.
  Then,
  \begin{equation*}
    \sqrt{n}(\hatbbeta^{(n)}    -    \bbeta^{\star})   \xrightarrow{d}
    \argmin_{\bu} V(\bu),
  \end{equation*}
  where, for $\bW \sim \mathcal{N}(\bzr, \sigma^2\bD)$,
  \begin{equation*}
    V(\bu) =
    \begin{cases}
      -2\bu^T       \bW      +
      \bu^T \bD \bu
      & \text{if $u_k = u_\ell$ for all $(k,\ell)\in\mathcal{A}^\star$} \\
      \infty & \text{otherwise.} \\
    \end{cases}
  \end{equation*}
\end{lemma}

\begin{proof}
  Let  $\hatbbeta^{(n)} = \bbeta^{\star}  + \frac{\bu_n}{\sqrt{n}}$  -- or
  equivalently  $\bu_n =  \sqrt{n}(\hatbbeta^{(n)}  - \bbeta^{\star})$  --
  where $\bbeta^{\star}$ is the true vector of parameters and $\bu_n =
  \argmin_{\bu\in\Rset^K} \Phi_n(\bu)$ with
  \begin{multline*}
    \Phi_n(\bu) =  \frac{1}{2} \sum_{k=1}^K
    n_k \left(y_i - (\beta^{\star}_k + \frac{u_k}{\sqrt{n}})\right)^2
    +  \lambda_n  \  \sum_{k\neq\ell}  w_{k\ell}^{\text{FA}}  \left| \beta^{\star}_k  -
    \beta^{\star}_\ell + \frac{u_k-u_\ell}{\sqrt{n}} \right|.
  \end{multline*}
  Note that $\bu_n$ is also the minimizer of $V_n(\bu) = \Phi_n(\bu) -
  \Phi(\bzr)$ which is written
  \begin{equation*}
    V_n(\bu) =  \sum_{k} \frac{n_k}{n} u_k^2 -  2 \sum_k \frac{n_k}{n}
    \varepsilon_k     +    \frac{\lambda_n}{\sqrt{n}}    \sum_{k,\ell}
    w_{k\ell}^\text{FA} \underbrace{\sqrt{n} \left( \left|
        \beta^\star_k      -     \beta_\ell^k     +      \frac{u_k     -
          u_\ell}{\sqrt{n}}\right| - \left| \beta^\star_k - \beta_\ell^\star\right|\right)}_{T_{k\ell}^{(n)}}.
  \end{equation*}
  Let us study the limiting  behavior of $V_n$.  The basic assumptions
  for  our fused-ANOVA  Problem  \eqref{eq:fused_anova}  are having  a
  design such  that $\lim_{n\to\infty}  n_k / n  = \rho_k$  and having
  i.i.d residuals with zero mean and common variance $\sigma^2$. Thus,
  the first  two terms  in $V_n$ respectively  converge to  a constant
  $\bu^T  \bD  \bu$  and  to   a  Gaussian  $\bW  =  \mathcal{N}(\bzr,
  \sigma^2\bD)$, where $\bD$ is a  $K$-diagonal matrix such as $D_{kk}
  = \rho_k$.  For the third  term, there are two possibilities: either
  $\beta^\star_k   =   \beta_\ell^\star$    or   $\beta^\star_k   \neq
  \beta_\ell^\star$.    In   other   words,  $(k,\ell)$   belongs   to
  $\mathcal{A}^\star$ or does not. First note that
  \begin{equation*}
    T_{k\ell}^{(n)}  \xrightarrow[]{n\to\infty}
    \begin{cases}
      |u_k - u_{\ell}| &
      \text{if } (k,\ell) \in\mathcal{A}^\star,\\
      (u_k  - u_{\ell})  \sign(\beta_k^\star  - \beta_{\ell}^\star)  &
      \text{otherwise}.
    \end{cases}
  \end{equation*}
  In words, this part of the third term converges to a finite constant
  in both situations  which is null as soon  as $u_k=u_\ell$.  Second,
  consider the  remaining part of  this third term which  involves the
  weights   $w_{k\ell}^\text{FA}$.     It   suffices   to    use   the
  $\sqrt{n}$-consistency       of       the       OLS       estimators
  $(\bar{y}_1,\dots,\bar{y}_K)$ coupled  with assumptions made  on the
  limiting behavior of $\lambda_n$ to see that
  \begin{equation*}
    \frac{\lambda_n}{\sqrt{n}}         w_{k\ell}^{\text{FA}}         =
    \frac{\lambda_n}{\sqrt{n}}                n_k               n_\ell
    \exp\set{-\alpha\sqrt{n}|\bar{y}_k-\bar{y}_\ell|} \to
    \begin{cases}
      \infty & \text{if } (k,\ell) \in\mathcal{A}^\star,\\
      0 & \text{otherwise} .
    \end{cases}
  \end{equation*}
  Application of  Slutsky's Lemma gives  the limiting behavior  of the
  third term in $V_n$ and we finally get $V_n(\bu)\to V(\bu)$ with $V$
  defined as in Lemma \ref{lem:asym_distr}.

  The  final  convergence  of  $\bu_n  \to_d  \argmin_{\bu}  V(\bu)$  is
  obtained by applying epi-convergence results of \cite{Geyer1994}.
\end{proof}

Turning back to the proof  of Theorem \ref{thm:oracle_prop}, just note
that the  unique minimizer  of the convex  function $V(\bu)$  in Lemma
\ref{lem:asym_distr}    is   $\bu^\star    =    \bD^{-1}   \bW    \sim
\mathcal{N}(\bzr,       \sigma^2       \bD^{-1})$      such       that
$u^\star_k=u^\star_\ell$ for all  $(k,\ell) \in \mathcal{A}^\star$ and
the asymptotic normality part is proved.
\\

We now proceed to the consistency in terms of support recovery. First,
concerning the elements of $\hatbbeta$  that should not fuse according
to the true  $\mathcal{A}^\star$, Lemma \ref{lem:asym_distr} indicates
that
\begin{equation*}
  \prob\left((k,\ell) \notin \hat{\mathcal{A}}_n | (k,\ell) \notin
    \mathcal{A}^\star\right) = 1- \prob\left(\hat{\beta}_{k}^{(n)} =
    \hat{\beta}^{(n)}_{\ell} | \beta_k^\star \neq \beta_\ell^\star\right) \to 1.
\end{equation*}

Second, regarding  elements of $\hatbeta$  that must fuse, we  need to
prove that
\begin{equation*}
  \prob\left((k,\ell) \in \hat{\mathcal{A}}_n | (k,\ell) \in
    \mathcal{A}^\star\right) = \prob\left(\hat{\beta}_{k}^{(n)} =
    \hat{\beta}^{(n)}_{\ell} | \beta_k^\star = \beta_\ell^\star\right) \to 1.
\end{equation*}
To this  end, we  proceed as in  \cite{viallon:hal-00813281} to  get a
contradiction by considering the  largest $\hat{\beta}_{k'}$ such that
$\hat{\beta}_k  \neq  \hatbeta_\ell$   even  though  $\beta_k^\star  =
\beta_\ell^\star$.  This can be done  by inspecting the KKT conditions
asymptotically.   In   the  univariate  case  and   for  $\Omega$  the
$\ell_1$-norm,   an   optimal   $\hatbeta$  verifies   the   following
subgradient equation for all $k=1,\dots,K$. This is written
\begin{equation}
  \label{eq:KKT_asym}
  \frac{n_k}{\sqrt{n}} (\hat{\beta}_k - \bar{y}_k) = \frac{\lambda_n}{\sqrt{n}} \sum_{\substack{\ell:
      \ell\neq k \\ (k,\ell) \in \mathcal{A}^\star}}
  w_{k\ell}^{\text{FA}} \tau_{k\ell} + \frac{\lambda_n}{\sqrt{n}}
  \sum_{\substack{\ell:
      \ell\neq k \\ (k,\ell) \notin \mathcal{A}^\star}}
  w_{k\ell}^{\text{FA}} \sign\left(\hat{\beta}_k-\hat{\beta}_\ell\right).
\end{equation}

Now, in  the first  term of  the right-hand  side, suppose  that there
exists at  least one  $\ell$ such  that $(k,\ell)\in\mathcal{A}^\star$
and  $\hat{\beta}_k  \neq \hat{\beta}_\ell$  simultaneously;  consider
$\hat{\beta}_{k'}$               with              $k'               =
\argmax_{\ell:(k,\ell)\in\mathcal{A}^\star}\{  \hat{\beta}_{\ell} \}$,
the  largest coefficients  verifying  these conditions:  we must  have
$\tau_{k'\ell} =  1$ for all $\ell$  such that $\hat{\beta}_{\ell}\neq
\hat{\beta}_{k'}$ and $\beta_{\ell}^\star = \beta_{k'}^\star$.  Now if
we sum equation \eqref{eq:KKT_asym} for all $\ell$ that are fused with
$k'$ we obtain:
\begin{multline}
  \label{eq:KKT_asym_2}
  \sum_{\ell | \hat{\beta}_{\ell} = \hat{\beta}_{k'}} \frac{n_{\ell}}{\sqrt{n}} (\hat{\beta}_{k'} - \bar{y}_{\ell}) =
	\frac{\lambda_n}{\sqrt{n}}
	\sum_{\ell | \hat{\beta}_{\ell} = \hat{\beta}_{k'}} \
	\sum_{\substack{k | (k,\ell) \in \mathcal{A}^\star \\ \cap \ \hat{\beta}_{k} \neq \hat{\beta}_{k'}}}
  	w_{k\ell}^{\text{FA}} \\
 	+ \frac{\lambda_n}{\sqrt{n}}
	\sum_{\ell | \hat{\beta}_{\ell} = \hat{\beta}_{k'}} \
	\sum_{\substack{\ell:\ell\neq k \\ (k,\ell) \notin \mathcal{A}^\star}}
  	w_{k\ell}^{\text{FA}} \sign\left(\hat{\beta}_k-\hat{\beta}_\ell\right).
\end{multline}

By Lemma \ref{lem:asym_distr} and  asymptotic normality, the left-hand
side in \eqref{eq:KKT_asym_2} converges to a $\mathcal{O}_P(1)$. Then,
the second term on the right-hand  side (that is, elements that should
not      fuse)      tends       to      $0$      since      $\lambda_n
w_{k\ell}^{\text{FA}}/\sqrt{n}\to                0$               when
$(k,\ell)\notin\mathcal{A}^\star$,  as  seen previously.   Finally  we
have :
\begin{equation*}
  \frac{\lambda_n}{\sqrt{n}}
	\sum_{\ell | \hat{\beta}_{\ell} = \hat{\beta}_{k'}} \
	\sum_{\substack{k | (k,\ell) \in \mathcal{A}^\star \\ \cap \ \hat{\beta}_{k} \neq \hat{\beta}_{k'}}}
  	w_{k\ell}^{\text{FA}}
  \xrightarrow[]{n\to\infty} \infty
\end{equation*}
which is in contradiction with the rest of the subgradient equation of
$\beta_{k'}$   since   we   recall   that  the   left-hand   side   is
$\mathcal{O}_P(1)$.  Therefore  we must have  $\prob\left((k,\ell) \in
  \hat{\mathcal{A}}_n\right)\to    1$    for   all    $(k,\ell)    \in
\mathcal{A}^\star$, which completes the  proof of the consistency part
in Theorem \ref{thm:oracle_prop}.


\end{document}